\documentclass[letterpaper,11pt]{article}

\usepackage{amsmath,amsfonts,amsthm,amssymb,color}
\usepackage{fancybox}
\usepackage[colorlinks=true]{hyperref}
\usepackage[procnumbered,ruled,vlined,linesnumbered]{algorithm2e}
\usepackage{float}
\usepackage{graphicx}
\usepackage{subcaption}
\usepackage{thmtools}
\usepackage{thm-restate}

\newtheorem{theorem}{Theorem}[section]

\newtheorem{lemma}[theorem]{Lemma}

\newtheorem{fact}[theorem]{Fact}

\newtheorem{definition}[theorem]{Definition}

\newtheorem{problem}{Problem}

\newenvironment{fminipage}%
  {\begin{Sbox}\begin{minipage}}%
  {\end{minipage}\end{Sbox}\fbox{\TheSbox}}

\newenvironment{algbox}[0]{\vskip 0.1in
\noindent 
\begin{fminipage}{6.3in}
}{
\end{fminipage}
\vskip 0.1in
}

\DontPrintSemicolon
\SetKw{KwAnd}{and}
\SetProcNameSty{textsc}
\SetFuncSty{textsc}
\SetKwInOut{Input}{Input}%\ \ \ \ }
\SetKwInOut{Output}{Output}

\def\pleq{\preccurlyeq}

\def\prob#1#2{\mbox{Pr}_{#1}\left[ #2 \right]}

\def\expec#1#2{{\mathbb{E}}_{#1}\left[ #2 \right]}

\def\defeq{\stackrel{\mathrm{def}}{=}}
\def\setof#1{\left\{#1  \right\}}
\def\sizeof#1{\left|#1  \right|}

\def\floor#1{\left\lfloor #1 \right\rfloor}

\def\union{\cup}

\def\eps{\epsilon}

\def\abs#1{\left|#1  \right|}

\newcommand\Tcal{\mathcal{T}}
\newcommand\Pcal{\mathcal{P}}

\newcommand\pth{\mathfrak{p}}

\newcommand\er{R_{\mathrm{eff}}}

\newcommand\bb{\boldsymbol{\mathit{b}}}

\newcommand\ee{\boldsymbol{\mathit{e}}}

\newcommand\xx{\boldsymbol{\mathit{x}}}

\renewcommand\AA{\boldsymbol{\mathit{A}}}
\newcommand\BB{\boldsymbol{\mathit{B}}}

\newcommand\LL{\boldsymbol{\mathit{L}}}

\renewcommand\SS{\boldsymbol{\mathit{S}}}

\newcommand\WW{\boldsymbol{\mathit{W}}}

\newcommand\Otil{\widetilde{O}}

\newcommand\SSi{\boldsymbol{\Sigma}}

\DeclareMathOperator*{\argmax}{arg\,max}

\newcommand{\kh}[1]{\left(#1\right)}

\newcommand{\NSTMaximize}{\textsc{NSTMaximize}}

\newcommand{\AddAbove}{\textsc{AddAbove}}
\newcommand{\GreedyTh}{\textsc{GreedyTh}}

\newcommand{\ApproxSchur}{\textsc{ApproxSchur}}
\newcommand{\EREst}{\textsc{EREst}}

\newcommand{\polylog}{{\rm poly}\log}

\newcommand{\ces}{Q}
\newcommand{\cen}{q}
\newcommand{\ses}{P}

\title{Maximizing the Number of Spanning Trees in a Connected Graph\footnote{Stacy Patterson is supported in part by NSF grants CNS-1553340 and CNS-1527287.}}

\date{}

\author{
    Huan Li\thanks{%Shanghai Key Lab of Intelligent Information Processing,
    School of Computer Science,
    Fudan University. \texttt{email:huanli16@fudan.edu.cn}}
    \and
    Stacy Patterson\thanks{Department of Computer Science,
        Rensselaer Polytechnic Institute. %\hspace{68pt}
        \texttt{email:sep@cs.rpi.edu}}
    \and
    Yuhao Yi\thanks{%Shanghai Key Lab of Intelligent Information Processing,
    School of Computer Science,
    Fudan University. \texttt{email:yhyi15@fudan.edu.cn}}
    \and
    Zhongzhi Zhang\thanks{%Shanghai Key Lab of Intelligent Information Processing,
    School of Computer Science,
    Fudan University. \texttt{email:zhangzz@fudan.edu.cn}}
}

\begin{document}

%\pagenumbering{gobble}

\maketitle

            %$\Otil( ( m + (n + q) \epsilon^{-2} ) \epsilon^{-1} )$%
            %time,

% or its direct acceleration via
            %nearly-linear time effective resistances estimation with
            %approximation ratio $(1 - \frac{1}{e} - \eps)$ in the exponent
            %and running time $\Otil((m + (n + \cen) \eps^{-2})k)$, where the latter's
            %running time could be superlinear or even
            %quadratic when $k$ is large.
            %no polynomial time algorithm can
            % unless P = NP.

\begin{abstract}	
    We study the problem of maximizing the number of spanning trees in
    a connected graph by adding at most $k$ edges from a given
    candidate edge set,
    a problem that
    has applications
    in domains including
    robotics, network science, and cooperative control.
    By Kirchhoff's matrix-tree theorem,
    this problem is equivalent to maximizing the determinant
    of an SDDM matrix.
    We give both algorithmic and hardness results for this problem:
    \begin{itemize}
        \item We give a greedy algorithm that, using submodularity,
            obtains an approximation ratio of $(1 - 1/e - \epsilon)$
            in the exponent of the number of spanning trees for any $\epsilon > 0$
            in time\footnote{The notation $\Otil$ hides $\polylog(n)$ factors.}
            $\Otil(m \epsilon^{-1} + (n + q) \epsilon^{-3})$, where $n$ is the number of vertices, and $m$ and $q$ are
            the number of edges in the original graph and
            the candidate edge set, respectively.
            Our running time
            is optimal with respect to the input size, up to logarithmic factors,
            and substantially improves upon
            the $O(n^3)$ running time of the previous proposed greedy algorithm, which has 
            an approximation ratio $(1 - 1/e)$ in the exponent.
            Notably, the independence of our running time of $k$ is novel,
            compared to
            conventional top-$k$ selections on graphs
            that usually run in $\Omega(mk)$ time.
            A key ingredient of our greedy algorithm is a routine for maintaining
            effective resistances under edge additions
            that is a hybrid of online and offline processing techniques; this routine
             may be of independent interest in areas including dynamic algorithms
            and data streams.
        \item We show the exponential inapproximability of this problem by
            proving that there exists a constant $c > 0$ such that
            it is $\mathbf{NP}$-hard to approximate the optimum number of spanning trees
            in the exponent within $(1 - c)$.
            By our reduction, the inapproximability of this problem
            can also be stated as  there exists a constant $d > 0$ such that
            it is $\mathbf{NP}$-hard to approximate the optimum number of spanning trees
            within $(1 + d)^{-n}$.
            Our inapproximability result follows from a reduction from
            the minimum path cover in undirected graphs,
            whose hardness again follows from the constant inapproximability
            of the Traveling Salesman Problem (TSP) with distances 1 and 2.
            Thus, the approximation ratio of our algorithm
            is also optimal up to a constant factor in the exponent.
            To our knowledge, this is the first hardness of approximation result for
            maximizing the number of spanning trees in a graph,
            or equivalently,
            maximizing the determinant of an SDDM matrix.
    \end{itemize}
\end{abstract}%

\newpage

%\pagenumbering{arabic}

\section{Introduction}

We study the problem of maximizing the number of spanning trees
in a weighted connected graph $G$ by adding at most $k$ edges
from a given candidate edge set.
By Kirchhoff's matrix-tree theorem~\cite{Kir47},
the  number of spanning trees in $G$ is equivalent to the
determinant of a minor of the graph Laplacian $\LL$. 
Thus, an equivalent problem is to maximize the determinant of a minor of $\LL$,
or, more generally, to maximize the determinant of an SDDM matrix.
The problem of maximizing the number of spanning trees, and the related problem
of maximizing the determinant of an SDDM matrix, have applications in a wide variety
of problem domains. We briefly review some of these applications below.

In robotics, the problem of maximizing the number of spanning trees has been applied
in graph-based Simultaneous Localization and Mapping (SLAM).
In  graph-based SLAM~\cite{TM06}, each vertex corresponds to a robot's pose or position,
and edges correspond to relative measurements between poses. The graph is used to estimate the most
likely pose configurations. 
Since measurements can be noisy, a larger number of measurements results in a more accurate estimate.
 The problem of  selecting which $k$ measurements to add to a SLAM pose graph to most improve the estimate 
 has been recast as a problem of selecting the $k$ edges to add to the graph that maximize the number of spanning trees~\cite{khosoussi2015good,KHD16,KSHD16a,KSHD16b}. 
We note that the complexity of the estimation problem increases with the number of measurements, and so 
 sparse, well-connected pose graphs are desirable~\cite{DK06}. 
Thus, one expects $k$ to be moderately sized with respect to the number of vertices.

% applications  the need maximal number of spanning trees:
In network science, the number of spanning trees has been studied as a measure of reliability
in communication networks, where reliability is defined as the probability
that every pair of vertices can communicate~\cite{M96}. 
Thus, network reliability can be improved by adding edges that most increase
the number of spanning trees~\cite{FL01}.
The number of spanning trees has also been used as a predictor of the spread of information in
social networks~\cite{BAE11}, with a larger number of spanning trees corresponding to better
information propagation.

%A solution for identifying the best \emph{single} edge to add was  developed in this context~\cite{FL01}. 
In the field of cooperative control, the log-number of
spanning trees  has been shown to capture the robustness of linear consensus algorithms.
Specifically, the log-number of spanning trees quantifies the network entropy, a measure of how well
the agents in the network maintain agreement when subject to external
stochastic disturbances~\cite{SM14,dBCM15,ZEP11}.
Thus, the problem of selecting which edges to add to the network graph to optimize robustness is
equivalent to the log-number of spanning trees maximization problem~\cite{SM18,ZSA13}.
%It has also been shown that, in tree networks with linear consensus dynamics,
%when a single edge can be added,  the steady-state variance of the deviation from
%consensus is minimized  by adding the edge that yields the graph with the most spanning trees~\cite{ZSA13}.  
Finally, the  log-determinant of an SDDM matrix has also been used directly as a measure of
controllability in more general linear dynamical systems~\cite{SCL16}. In this paper, we provide an approximation algorithm to maximize the log-number of spanning trees of a connected graph by adding edges.

% optimal experiment design
   %     D-optimality in experiment design~\cite{C78,C81,P93}, and applications in sensor networks \cite{BOYD}

% prior work on maximizing number of spanning trees
%We also note that the problem of identifying  graphs with the maximal number of spanning trees,
%from among all graphs with a $n$ vertices and $m$ edges,
%has received significant attention~\cite{S74,C81,BLXS91,W94,K96,GM97,PBS98,P02}.
%Solutions to this problem have been found for specific classes of graphs, for example,
%graphs having up to $n+2$ edges~\cite{BLXS91} and $n+3$ edges~\cite{W94}.
%Of particular note, a regular complete multipartite graph has been shown to have the
%maximal number of spanning trees from among all simple graphs with the same number
%of vertices and edges~\cite{C81}.

%Survey of spanning trees: https://link.springer.com/article/10.1007/s00373-010-0973-2

\subsection{Our Results}\label{sec:res}

Let $G = (V,E)$ denote an undirected graph  
with $n$ vertices and $m$ edges,
and let $w : E \to \mathbb{R}^+$ denote the edge weight function.
For another graph $H$ with edges supported on a subset of $V$,
we write ``$G$ plus $H$'' or $G + H$ to denote the graph obtained by
adding all edges in $H$ to $G$.

Let $\LL$ denote the Laplacian matrix of a graph $G$.
The effective resistance $\er(u,v)$ between two vertices $u$ and $v$
%or for an edge $e = (u,v)$,
is given by
\begin{align*}
    \er(u,v) \defeq (\ee_u - \ee_v)^T \LL^\dag (\ee_u - \ee_v),
\end{align*}
where $\ee_u$ denotes the $u^{\mathrm{th}}$ standard basis vector
and $\LL^\dag$ denotes the Moore-Penrose inverse of $\LL$.

The weight of a spanning tree $T$ in $G$
is defined as
\begin{align*}
    w(T) \defeq \prod_{e\in T} w(e),
\end{align*}
and the weighted number of spanning trees in $G$ is defined
as the sum of  the weights of all spanning trees,
denoted by $\Tcal(G) \defeq \sum\nolimits_{T} w(T)$.
By Kirchhoff's matrix-tree theorem~\cite{Kir47}, the weighted number of spanning trees
equals the determinant of a minor of the graph Laplacian:
\begin{align*}
    \Tcal(G) = \det\kh{\LL_{1:n-1,1:n-1}}.
\end{align*}

In this paper, we study the problem of maximizing the weighted number of spanning
trees in a connected graph by adding at most $k$ edges from a given candidate edge
set. We give a formal description of this problem below.
\begin{problem}[Number of Spanning Trees Maximization (NSTM)]\label{prob:ecard}
    Given a connected
    undirected graph $G = (V, E)$,
    an edge set $\ces$ of $\cen$ edges, % such that $\ces \intersect E = \emptyset$,
    an edge weight function $w : (E \union \ces) \to \mathbb{R}^+$,
    and an integer $1 \leq k \leq \cen$,
    add at most $k$ edges from $\ces$ to $G$ so that the
    weighted number of spanning trees in $G$ is maximized.
    Namely, the goal is to find a set $\ses\subseteq \ces$ of at most $k$ edges such that
    \begin{align*}
        \ses \in \argmax_{S\subseteq \ces,\sizeof{S}\leq k} \Tcal(G + S).
    \end{align*}
\end{problem}

\paragraph{Algorithmic Results.}{Our main algorithmic result is solving Problem~\ref{prob:ecard}
with an approximation factor of $(1 - \frac{1}{e} - \eps)$ in the exponent
of $\Tcal(G)$ in nearly-linear time, % using the submodularity of $\log \Tcal(G)$,
which can be described by the following theorem:
\begin{restatable}[]{theorem}{thmalgo}
    There is an algorithm $\ses = \NSTMaximize(G,\ces,w,\eps,k)$, which
    takes a connected graph $G = (V,E)$ with $n$ vertices and $m$ edges,
    %such that $\Tcal(G) \geq 1$,
    an edge set $\ces$ of $\cen$ edges, %such that $\ces\intersect E = \emptyset$,
    an edge weight function $w : (E \union \ces) \to \mathbb{R}^+$,
    a real number $0 < \eps \leq 1/2$,
    and an integer $1\leq k\leq \cen$,
    and returns an edge set $\ses \subseteq \ces$ of at most $k$ edges
    in time $\Otil((m + (n + \cen) \eps^{-2})\eps^{-1})$.
    With high probability, the following statement holds:
    \begin{align*}
        \log \frac{\Tcal(G + \ses)}{\Tcal(G)} \geq
        \kh{1 - \frac{1}{e} - \eps} \log \frac{\Tcal(G + O)}{\Tcal(G)},
    \end{align*}
    where $O \defeq \argmax\nolimits_{S \subseteq \ces, \sizeof{S}\leq k} \Tcal(G + S)$
    denotes an optimum solution.
    \label{thm:DetMaximize}
\end{restatable}
%Note that the only reason that we need $\Tcal(G)$ to be at least $1$
%is to ensure that the logarithm of the number of spanning trees
%is non-negative, and thus the approximation ratio is well-defined.

The running time of $\NSTMaximize$ is independent of the number $k$
of edges to add to $G$, and it is optimal with respect to the input size
up to logarithmic factors. This running time substantially improves upon
the previous greedy algorithm's $O(n^3)$ running time~\cite{KSHD16b}, or
the $\Otil((m + (n + \cen)\eps^{-2})k)$
running time of its direct acceleration via
fast effective resistance approximation~\cite{SS11,DKP+17} ,
where the latter becomes quadratic when $k$ is $\Omega(n)$.
Moreover, the independence of our running time of $k$ is
novel,
comparing to conventional top-$k$ selections on graphs %relying on submodularity
that usually run in $\Omega(mk)$ time,
such as the ones in~\cite{BBCL14,Yos14,MTU16,LPS+18}.
We briefly introduce these top-$k$ selections
in Section~\ref{sec:related}.
%especially
%for large $k$.

A key ingredient of the algorithm $\NSTMaximize$ is
a routine $\AddAbove$ that,
%a data structure which,
%which stores a connected graph,
%and supports a specific modifying operation:
given a sequence of edges and a threshold, % ordering of the edges in $\ces$,
\textit{sequentially} adds to the graph any edge whose effective resistance
(up to a $1 \pm \eps$ error)
is above the threshold at the time the edge is processed.
The routine $\AddAbove$ runs in nearly-linear time
in the total number of edges in the graph and the edge sequence.
%performs one pass of \textit{sequential} insertions
%for edges with effective resistances above a given threshold.
%where each pass runs in nearly-linear time.
%\textit{sequentially} picks edges whose
%effective resistances are above a given threshold in nearly-linear time.
%This data structure is motivated by the divide-and-conquer used
%to generate random spanning trees in~\cite{DKP+17}.
The performance of $\AddAbove$ is characterized in the following lemma:
\begin{restatable}[]{lemma}{lemalgo}
    There is a routine
    $\ses = \AddAbove(G, (u_i,v_i)_{i=1}^{\cen}, w, th, \eps, k)$, which takes
    a connected graph $G = (V,E)$ with $n$ vertices and $m$ edges,
    %such that $\Tcal(G) \geq 1$,
    an edge sequence $(u_i,v_i)_{i=1}^{\cen}$,
    an edge weight function $w : (E \union (u_i,v_i)_{i=1}^{\cen}) \to \mathbb{R}^+$,
    %$w_f : \setof{f_i}_{i=1}^{\cen} \to \mathbb{R}^+$,
    real numbers $th$ and $0 < \eps \leq 1 / 2$, and an integer $k$,
    and performs a sequential edges additions to $G$
    and returns the set $P$ of edges that have been added with $\sizeof{P}\leq k$.
    The routine $\AddAbove$ runs in time $\Otil(m + (n + \cen) \eps^{-2})$.
    With high probability,
    there exist $\kh{\hat{r}_i}_{i=1}^{\cen}$ such that
    $\AddAbove$ has the same return value as the following procedure,
    in which
    \[
        (1 - 2\eps) \er^{G^{(i)}}(u_i,v_i) \leq \hat{r}_i \leq (1 + 2\eps) \er^{G^{(i)}}(u_i,v_i)
    \]
    holds for all $i=1,2,\ldots,\cen$:
    %in terms of the return value: % and the modifications to $\GraphER .G$:
%    \begin{center}
%        \begin{enumerate}
%            \item[] \texttt{For $i=1$ to $q$}
%        \end{enumerate}
%    \end{center}
    \begin{algbox}
        $G^{(1)} \gets G$ \\
        \For{$i = 1$ {\rm to} $\cen$}{
            \If{$w(u_i,v_i)\cdot \hat{r}_i \geq th$ {\rm and} $k > 0$}{
                $G^{(i+1)} \gets G^{(i)} + (u_i,v_i)$, $k \gets k - 1$
            }
            \Else{
                $G^{(i+1)} \gets G^{(i)}$
            }
            %{\rm If} $k > 0$ {\rm and} $w(u_i,v_i)\cdot \hat{r}_i \geq th$,
            %{\rm set} $G^{(i+1)} \gets G^{(i)} + (u_i,v_i)$ {\rm and} $k \gets k - 1$;
            %{\rm Otherwise, set} $G^{(i+1)} \gets G^{(i)}$.
        }
        {\rm Return the set of edges in} $G^{(\cen+1)}$ {\rm but not in} $G$.
        %\KwRet {\rm }
        %{\rm Set} $\GraphER .G \gets G^{(q+1)}$ {\rm and return} $k$.
    \end{algbox}
    \label{lem:ds}
\end{restatable}

%We sketch the idea behind our algorithm below.
%
%Let $G + (u,v)$ denote the graph obtained from $G$ by adding edge $(u,v)$.
%By matrix determinant lemma~\cite{Har97},
%the weighted number of spanning trees in $G + (u,v)$ can be written
%in terms of $\Tcal(G)$ as
%\begin{align*}
%    \Tcal(G + (u,v)) = \kh{1 + w(u,v)\er(u,v)} \Tcal(G),
%\end{align*}
%where
%$
%    \er(u,v) \defeq (\ee_u - \ee_v) \LL^\dag (\ee_u - \ee_v)
%$
%is the effective resistance between vertices $u$ and $v$,
%in which $\LL^\dag$ denotes the Moore-Penrose inverse of $\LL$
%and $\ee_u$ denote the $u^{\mathrm{th}}$ standard basis vector.
%The submodularity of $\log \Tcal(G)$ then follows immediately from
%Rayleigh's monotonicity law~\cite{Sto87}.
%For maximizing a monotone submodular function faster,
%\cite{BV14} gives an algorithm which obtains
%a $(1 - \frac{1}{e} - \eps)$-approximation using nearly-linear objective function
%evaluations.
%However, since estimating $\Tcal(G)$ (or $\log \Tcal(G)$) is expensive,
%directly applying that algorithm here still makes the running time
%at least quadratic.
%
%A key step of the algorithm in~\cite{BV14}, interpreted in our setting,
%is to \textit{sequentially} pick edges whose effective resistances are above
%a given threshold.

The routine $\AddAbove$ can be seen 
as a hybrid of online and offline
processing techniques. The routine is provided a specific edge sequence as input, 
as is typical in offline graph algorithms. However, the routine does not know what operation should be performed on an edge 
(i.e., whether the edge should be added to the graph) until the edge is processed, in an online fashion. 
The routine thus has to alternately update the graph and query effective resistance. This routine may be of independent interest in areas including dynamic algorithms
and data streams.
%The routine $\AddAbove$ can be seen as an interpolation
%\hl{Is this word used properly?} \sep{See my comments is Section 1.2} 
%\hl{will explain in email}
%between online
%and offline.
%For the offline part,
%it is given a specific edge sequence in the input;
%for the online part, it does not know whether to add an edge
%until it processes that specific edge, and thus has to
%update the graph and query effective resistance alternately.
%Thus,
%this routine may be of independent interest in areas like dynamic algorithms
%and data streams\hl{Can we say this?}.

\paragraph{Hardness Results.}{To further show that the approximation ratio of the algorithm $\NSTMaximize$ is also
nearly optimal, we prove the following theorem, which indicates that
Problem~\ref{prob:ecard} is exponentially inapproximable:

\begin{theorem}\label{thm:hardness}
    There is a constant $c > 0$ such that
    given an instance of Problem~\ref{prob:ecard},
    it is $\mathbf{NP}$-hard to find an edge set $\ses \subseteq \ces$ with $\sizeof{\ses} \leq k$
    satisfying
    \[
        \log \frac{\Tcal(G + \ses)}{\Tcal(G)} > (1 - c) \cdot \log \frac{\Tcal(G + O)}{\Tcal(G)},
    \]
    %no polynomial time algorithm can
    %approximate \[\log \frac{\Tcal(G + O)}{\Tcal(G)}\] within $(1-c)$ unless $\rm{P=NP}$,
    where $O$ is an optimum solution defined in Theorem~\ref{thm:DetMaximize}.
\end{theorem}

The proof of Theorem~\ref{thm:hardness} follows by Lemma~\ref{soundness:lemma}. By the same lemma, we can also state the inapproximability of Problem~\ref{prob:ecard}
using the following theorem:

\begin{theorem}\label{thm:hardness2}
    There is a constant $d > 0$ such that
    given an instance of Problem~\ref{prob:ecard},
    it is $\mathbf{NP}$-hard to find an edge set $\ses \subseteq \ces$ with $\sizeof{\ses} \leq k$
    satisfying
    \[
        \Tcal(G + \ses) > \frac{1}{(1 + d)^n} \cdot \Tcal(G + O),
    \]
    %no polynomial time algorithm can
    %approximate \[\log \frac{\Tcal(G + O)}{\Tcal(G)}\] within $(1-c)$ unless $\rm{P=NP}$,
    where $O$ is an optimum solution defined in Theorem~\ref{thm:DetMaximize}.
\end{theorem}

Theorem~\ref{thm:hardness} implies that the approximation ratio of $\NSTMaximize$ is optimal up to a constant
factor in the exponent.
To our knowledge, this is the first hardness of approximation result
for maximizing the number of spanning trees in a graph, or equivalently,
maximizing the determinant of an SDDM matrix (a graph Laplacian minor).

In proving Theorem~\ref{thm:hardness}, we give a reduction from 
the minimum path cover in undirected graphs, whose hardness
follows from the constant inapproximability of the traveling salesman
problem (TSP) with distances $1$ and $2$.
The idea behind our reduction is to consider
a special family of graphs,
each graph from which
equals a star graph plus an arbitrary graph supported on its leaves.
%that can be obtained
%by adding a star
%(i.e. adding a dummy vertex connecting to all existing vertices)
%to one of its subgraphs.
Let $H = (V,E)$ be a graph %from such family
equal to a star $S_n$ plus
its subgraph $H[V'] = (V',E')$ supported on $S_n$'s leaves. % for a vertex set $V'$.
We can construct an instance of Problem~\ref{prob:ecard} from $H$ by letting
the original graph, the candidate edge set, and the number of edges to add
be, respectively
\begin{align*}
    G &\gets S_n, \hspace{10pt} Q \gets E', \hspace{10pt} k \gets \sizeof{V'} - 1.
\end{align*}
We give an example of such an instance in Figure~\ref{starPlusH:fig}.

\begin{figure}
    \centering
    \includegraphics[width=.6\linewidth]{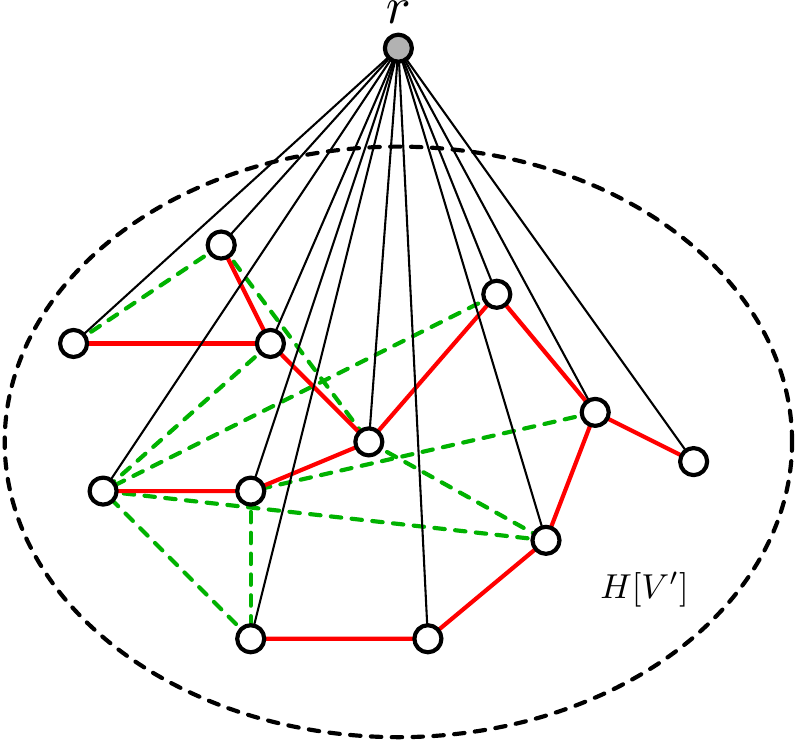}
    \caption{An instance of Problem~\ref{prob:ecard}
        constructed from $H$,
        which equals a star graph plus $H[V']$ supported on its leaves.
        Here, $r$ is the central vertex of the star.
        All red edges and green edges belong to the candidate edge set,
    where red edges denotes a possible selection with size $\sizeof{V'} - 1$.}
    \label{starPlusH:fig}
\end{figure}

%\begin{align*}
%    G &\gets S_n, \\
%    %\hspace{9pt}
%    Q &\gets E', \\
%    %\hspace{9pt
%    k &\gets \sizeof{V'} - 1.
%\end{align*}

%the original graph to
%$G = S_n$,
%the candidate edge set to $Q = E'$,
%and the number of edges to add to $k = \sizeof{V'} - 1$.
We then show in the following lemma that for two such instances whose $H[V']$s have respective
path cover number $1$ and $\Omega(\sizeof{V'})$,
the optimum numbers of spanning trees differ by a constant factor
in the exponent.

\begin{restatable}[]{lemma}{soundness}
	\label{soundness:lemma}
    Let $H=(V, E)$ be an unweighted graph equal to a star $S_n$ plus $H$'s subgraph 
    $H[V']=(V', E')$ supported on $S_n$'s leaves.
	For any constant $0<\delta<1$, there exists an absolute constant $c>0$ such that, if
	$H[V']=(V', E')$ does not have
	any path cover $\Pcal$ with $|\Pcal|<\delta n$, then
    \[
    \log \Tcal(S_n+\ses) \leq (1-c)\cdot \log \Tcal\kh{F_n}
\]
    holds for any $\ses\subseteq E'$ with $\sizeof{\ses}\leq n-1$.
    Here $F_n$ is a fan graph with $n-1$ triangles
    (i.e., a star $S_n$ plus a path supported on its leaves).
\end{restatable}

We remark that our reduction uses only simple graphs with all edge weights being 1.
Thus, Problem~\ref{prob:ecard} is exponentially inapproximable even for unweighted graphs
without self-loops and multi-edges.
}

\subsection{Ideas and Techniques}\label{sec:ideas}

%In the algorithm $\NSTMaximize$,
%to avoid searching for the edge with highest effective resistance,
%we maintain a geometrically decreasing threshold
%and sequentially pick any edge with effective resistance above the threshold.
%To 

\paragraph{Algorithms.}{
By the matrix determinant lemma~\cite{Har97},
the weighted number of spanning trees multiplies by
\begin{align*}
    1 + w(u,v) \er(u,v)
\end{align*}
upon the addition of edge $(u,v)$.
Then, the submodularity of $\log \Tcal(G)$ follows immediately
by Rayleigh's monotonicity law~\cite{Sto87}.
This indicates that
one can use a simple greedy algorithm~\cite{NWF78}
that
picks the edge with highest effective resistance iteratively for $k$ times
to achieve a $(1 - \frac{1}{e})$-approximation.
By computing effective resistances in nearly-linear time~\cite{SS11,DKP+17},
one can implement this greedy algorithm in $\Otil((m + (n + \cen) \eps^{-2})k)$ time
and obtain a $(1 - \frac{1}{e}-\eps)$-approximation.
%and picking the edge with highest effective resistance $k$ times,
%where $\cen$ is the number of edges in the candidate edge set.
To avoid searching for the edge with maximum effective resistance,
one can invoke another greedy algorithm proposed in~\cite{BV14},
which maintains a geometrically decreasing threshold
and sequentially picks any edge with effective resistance above the threshold.
However,
since the latter part of this greedy algorithm requires the recomputation of effective resistances
after each edge addition, it still needs $\Otil((m + \cen)k)$ running time.
Thus, our task reduces to performing the sequential updates faster.

We note that for a specific threshold,
the ordering in which we perform the sequential updates
does not affect our overall approximation.
Thus, by picking an arbitrary ordering of the edges in candidate set $\ces$,
we can  transform this seemingly  online task of processing edges sequentially into an online-offline hybrid setting.
%we can turn this task into a setting which is an interpolation
%\hl{Is this word used properly?}
%between online and offline.
% \sep{I'm not sure what you mean by interpolation. I think my confusion arises because I do not understand what you mean by ``offline'' and ``online''.
%To me, ``offline'' indicates some kind of pre-processing, and I'm not sure that makes sense in this context. Could you provide more detail about the intended meaning of offline here? Hopefully, with more information, I can answer your questions about the word.} \hl{We'll explain this in the email.}
While we do not know whether to add an edge
until the time we process it, we do know the order
in which the edges will be processed.
%We then gradually turn the online part into offline by performing 
We perform
divide-and-conquer on the edge sequence,
while alternately querying effective resistance and updating the graph.
The idea is that if we are dealing with a short interval of the edge sequence, instead of working with the entire graph,
we can work with a graph with size proportional to the length of the interval that preserves the effective resistances of the edges in the sequence.
%we could work with an equivalent graph \sep{equivalent graph to what? Is this graph equivalent to the input graph with respect to the effective resistances of the edges in the sequence?} \hl{Yes.} with size proportional to the length
%of that interval instead of the entire graph.
As we are querying effective resistances for candidate edges, 
%\sep{Could you explain what you mean by `querying'? Is this  computing the effective resistance for a candidate edge?} \hl{Yes.}
the equivalent graph for an interval can be obtained by
taking  the Schur complement
onto endpoints of the edges in it.
And, this can be done in nearly-linear time in the graph size
using the approximate
Schur complement routine in~\cite{DKP+17}.

Specifically, in the first step of divide-and-conquer,
we split the edge sequence $(f_i)_{i=1}^{\cen}$ into two halves
\[
    f^{(1)} \defeq f_1,\ldots, f_{\floor{\cen / 2}}
    \qquad
    \text{and}
    \qquad
    f^{(2)} \defeq f_{\floor{\cen / 2} + 1},\ldots,f_{\cen}.
\]
We note the following:
\begin{enumerate}
    \item Edge additions in $f^{(2)}$ do not affect effective
        resistance queries in $f^{(1)}$.
    \item An effective resistance query in $f^{(2)}$ is affected
        by
        \begin{enumerate}
            \item \label{2a} edge additions in $f^{(1)}$, and
            \item \label{2b} edge additions in $f^{(2)}$ which
                are performed before the specific query.
        \end{enumerate}
\end{enumerate}
Since $f^{(1)}$ is completely independent of $f^{(2)}$,
we can handle queries and updates in $f^{(1)}$ by performing recursion to
its Schur complement.
%by addressing the two parts of influence to its queries as follows.
We then note that edge additions
in $f^{(1)}$ are performed entirely before queries in $f^{(2)}$,
and thus can be seen as offline modifications to $f^{(2)}$.
Moreover, all queries in $f^{(2)}$ are affected by
the same set of modifications in $f^{(1)}$.
We thus address the total contribution of $f^{(1)}$ to $f^{(2)}$
%(i.e.~(\ref{2a}))
by computing Schur complement onto $f^{(2)}$ in the graph updated
by edge additions in $f^{(1)}$.
In doing so,
we have addressed~(\ref{2a}) for all queries in $f^{(2)}$,
and thus have made $f^{(2)}$ independent of $f^{(1)}$.
%and thus made
%queries in $f^{(2)}$ only affected by~(\ref{2b}).
%Since~(\ref{2b}) is independent of $f^{(1)}$,
%we make $f^{(2)}$ also completely independent of $f^{(1)}$,
%which follows by that~(\ref{2b}) is entirely within $f^{(2)}$.
This indicates that we can process 
%\sep{by handle, do you mean compute effective resistances?} 
%\hl{compute their effective resistances, and update the graph
 %(add the edge) if the effective resistance is larger than
 % some threshold.} 
$f^{(2)}$ by also performing recursion to its Schur complement.
We keep recursing until the interval only contains one edge,
where we directly query the edge's effective resistance
and decide whether to add it to the graph.
Essentially, our algorithm computes the effective resistance of each
edge with an elimination of the entire rest of the graph, while
heavily re-using previous eliminations.
This gives a
nearly-linear time routine for performing sequential updates.
Details for this routine can be found in Section~\ref{sec:seq}.

\paragraph{Hardness.}{
    A key step in our reduction is to show
    the connection between the minimum path cover
    and Problem~\ref{prob:ecard}.
    To this end,
    we consider an instance of Problem~\ref{prob:ecard}
    in which $G$ is a star graph $S_n$ with $n$ leaves,
    the candidate edge set $\ces$ forms an underlying graph supported on $S_n$'s leaves,
    and the number of edges to add equals $k = n - 1$.
    We show that for two instances
    whose underlying graphs have respective path cover number
    $1$ and $\Omega(n)$,
    their optimum numbers of spanning trees differ exponentially.

    Consider any set $\ses$ that consists of $n - 1$ edges from $\ces$,
    and any path cover $\Pcal=\{\pth_1, \pth_2,\dots, \pth_{t} \}$ of the underlying graph
    using only edges in $\ses$.
    Clearly $t$ is greater than or equal to
    the minimum path cover number of the underlying graph.
    If $\ses$ forms a Hamiltonian path $\Pcal^\ast$ in the underlying graph,
    $\Tcal(S_n + P)$ can be explicitly calculated~\cite{MEM14}
    and equals
    \begin{align}\label{eq:fannst}
        \Tcal(S_n + P) = \Tcal(S_n + \Pcal^\ast)=
        \frac{1}{\sqrt{5}}\kh{\kh{\frac{3+\sqrt{5}}{2}}^{n} -\kh{\frac{3-\sqrt{5}}{2}}^{n}}.
    \end{align}
    
    When the path cover number of the underlying graph is more than $1$,
    $\Tcal(S_n+P)$ can be expressed by a product of a sequence of effective resistances and
    $\Tcal(G+\Pcal)$. Specifically,
    for an arbitrary ordering
    $(u_i,v_i)_{i=1}^{t - 1}$ of edges in
    $\ses$ but not in the path cover $\Pcal$,
    %$\ses \setminus \Pcal$,
    we define a graph sequence $G^{(1)},\ldots,G^{(t)}$
    by
    \begin{align*}
        & G^{(1)} \hspace{10pt} \defeq G + \Pcal, \\
        & G^{(i+1)} \defeq G^{(i)} + (u_i, v_i)\hspace{5pt} \text{for $i=1,\ldots,t-1.$}
    \end{align*}
    By the matrix determinant lemma,
    we can write the number of spanning trees in $G^{(t)}$ as
    \begin{align}\label{eq:nstprod}
        \Tcal\kh{G^{(t)}} =
        \Tcal\kh{G^{(1)}} \cdot
        \prod_{i=1}^{t - 1} \kh{1 + \er^{G^{(i)}}(u_i,v_i)}.
    \end{align}
    Note that we omit edge weights here since we are dealing with
    unweighted graphs.
    
    Let %$\Pcal=\{\pth_1, \pth_2,\dots, \pth_{t} \}$,
    $l_i \defeq \sizeof{\pth_{i}}$ be the number of edges in path $\pth_{i}$.
    Since all paths $\pth_i\in \Pcal$ are disjoint,
    $\Tcal\kh{G^{(1)}}$ can be expressed as
    \begin{align}
    	\Tcal\kh{G^{(1)}}
        = &\prod_{i=1}^{t} \Tcal\kh{S_{l_i+1} + \pth_i} \notag \\
        = &\prod_{i=1}^{t}\kh{\frac{1}{\sqrt{5}}
        \kh{\kh{\frac{3+\sqrt{5}}{2}}^{l_i+1} -\kh{\frac{3-\sqrt{5}}{2}}^{l_i+1}} },
        \label{eq:nstg1}
    \end{align}
    where
    $S_{l_i+1}$ denotes the star graph with $l_i + 1$ leaves,
    and the second equality follows from~(\ref{eq:fannst}).

    When the path cover number of the underlying graph is at least $\Omega(n)$,
    we show that the number of spanning trees in $S_n + P$ is exponentially smaller
    than $\Tcal(S_n + \Pcal^\ast)$.
    %for two reasons:
    %\begin{enumerate}
    %    \item \label{rs1} $\Omega(n)$ edges in $\ses$ but not in $\Pcal$;
    %    \item \label{rs2} $\Omega(n)$ paths in $\Pcal$ with $O(1)$ lengths,
    %        whose existence follows from Markov's inequality.
    %\end{enumerate}
    %\yhy{Still need to polish}
    Let $\Pcal_1$ denote the set of paths in $\Pcal$ with $O(1)$ lengths.
    Let $t_1$ denote the number of paths in $\Pcal_1$.
    %$\gamma_1$ denote the total path lengths in $\Pcal_1$, and $n_1$ denote 
    %the total number of vertices covered in $\Pcal_1$.
    Then, %\ref{rs2} makes 
    $\Tcal(S_n + P)$ is exponentially smaller
    due to the following reasons. First, by~(\ref{eq:fannst}) and~(\ref{eq:nstg1}), $\Tcal(G^{(1)})$
    is less than 
    %the number of spanning trees in
    $\Tcal(S_n + \Pcal^\ast)$ by at least a multiplicative factor of
    $(\sqrt{5})^{t-1}\kh{1+\theta}^{t_1}$ for some constant $\theta>0$.
    Second, the effective resistances between endpoints of the $t-1$ edges in $\ses$ but not in
    the path cover $\Pcal$
    are less than $\sqrt{5} - 1$
    and hence cannot compensate for the $(\sqrt{5})^{t-1}$ factor.
    Third, $t_1$ is at least $\Omega(n)$ by Markov's inequality,
    which ensures that the factor $(1 + \theta)^{t_1}$ is exponential.
    %this gap cannot be 
    %compensated by adding $n_1-l_1-1$ edges that are not in $\Pcal$.
    %Here $\theta>0$ is a constant.
    %in a fan graph with the same edge number.
    %which can be verified using the formula in~\cite{MEM14}.
    %To give the intuition behind~\ref{rs1},
%    
%    
%    This can be shown by the following analyses.
%    We write $\Tcal(S_n + \ses)$ in a product of effective resistances.
%    Specifically,
%    for an arbitrary ordering
%    $(u_i,v_i)_{i=1}^{\sizeof{\Pcal} - 1}$ of edges in $\ses \setminus \Pcal$,
%    we define a graph sequence $G^{(1)},\ldots,G^{(\sizeof{\Pcal})}$
%    by
%    \begin{align*}
%        & G^{(1)} \hspace{10pt} \defeq G + \Pcal, \\
%        & G^{(i+1)} \defeq G^{(i)} + (u_i, v_i)\hspace{5pt} \text{for $i=1,\ldots,\sizeof{\Pcal}-1.$}
%    \end{align*}
%    By matrix determinant lemma,
%    we can write the number of spanning trees in $G^{(\sizeof{\Pcal})}$ as
%    \begin{align}\label{eq:nstprod}
%        \Tcal\kh{G^{(\sizeof{\Pcal})}} =
%        \Tcal\kh{G^{(1)}} \cdot
%        \prod_{i=1}^{\sizeof{\Pcal} - 1} \kh{1 + \er^{G^{(i)}}(u_i,v_i)}.
%    \end{align}
%    Note that we omit edge weights here since we are dealing with
%    unweighted graphs.
%    Then, by upper bounding effective resistances
%    for edges not in $\Pcal$,
%    %between vertices in different paths of $\Pcal$,
%    we show that $\Omega(n)$ factors in~(\ref{eq:nstprod}) are strictly
%    bounded away from the factor $\sqrt{5}$ in~(\ref{eq:fannst}).
%    %\yhy{Is this explanation convincing?}
    This leads to the exponential drop of
    $\Tcal(S_n + \ses)$.
    We defer our proof details to Section~\ref{hardness::sec}.

    %Intuitively, this relation can be illustrated by considering two extreme
    %examples: a fan graph $F_n$ with $n - 1$ triangles
    %(i.e., a star $S_{n}$ plus a path supported on its leaves), and a graph
    %equal to a star $S_n$ plus another star $S_{n-1}$ supported on its leaves.
    %We draw these two graphs for the case $n = 4$ in Figure\todo{}.
    %It is clear that the respective minimum path cover number of the two
    %underlying graphs
    %is $1$ and $n - 2$.
    %On the other hand,
    %the respective number of spanning trees in the two graphs
    %is
    %\begin{align}\label{eq:fannst}
    %\frac{1}{\sqrt{5}}\kh{\kh{\frac{3+\sqrt{5}}{2}}^{n} -\kh{\frac{3-\sqrt{5}}{2}}^{n}}
    %\end{align}
    %and $(n + 1) 2^{n-2}$.
    %Here the first number of spanning trees follows from~\cite{MEM14},
    %and the second follows by a direct calculation using matrix tree theorem.
    %One can observe a exponential difference between 
}

\subsection{Related Work}\label{sec:related}

%\paragraph{Applications}{
%    \begin{enumerate}
%        \item Robotics: the number of spanning trees
%            measures the tree-connectivity in sensor
%            networks~\cite{KHD16,KSHD16a,KSHD16b}
%        \item Consensus (control):
%            the number of spanning trees measures the robustness of
%            consensus algorithms\cite{SM14,dBCM15},
%            Section IV of~\cite{ZEP11}.
%        \item Use of determinant as performance measure in linear dynamical systems network~\cite{SM18,SCL16}
%        \item Reliability of networks: the number of spanning trees
%            measures the reliability of a network~\cite{FL01}
%    \end{enumerate}
%}

\paragraph{Maximizing the Number of Spanning Trees.}{
   There has been limited previous algorithmic study of the problem of maximizing the number of spanning trees in a graph.

    Problem~\ref{prob:ecard} was also studied in~\cite{KSHD16b}.
    This work proposed a greedy algorithm which,
    by computing effective resistances exactly,
    achieves an approximation factor of $(1 - \frac{1}{e})$
    in the exponent of the number of spanning trees
    in $O(n^3)$ time.
    As far as we are aware, the hardness of Problem~\ref{prob:ecard} has not been studied in any previous work.
    
    A related problem of,
    given $n$ and $m$,
    identifying graphs
    with $n$ vertices and $m$ edges
    that have  the maximum number of spanning trees has been studied.
    However, most solutions found to this problem are
    for either sparse graphs with $m = O(n)$ edges~\cite{BLXS91,W94},
    or dense graphs with $m = n^2 - O(n)$ edges~\cite{S74,K96,GM97,PBS98,P02}.
    Of particular note, a regular complete multipartite graph has been shown to have the
    maximum number of spanning trees from among all simple graphs with the same number
    of vertices and edges~\cite{C81}.
    %\cite{K96}
}

\paragraph{Maximizing Determinants.}{
    Maximizing the determinant of a positive semidefinite
    matrix (PSD) is a problem that has been extensively studied
    in the theory of computation.
    For selecting a principle minor of a PSD with the maximum determinant
    under certain constraints,
    ~\cite{Kha95,CM09,SEFM15,NS16,ESV17} gave algorithms for approximating the optimum
    solution.
    \cite{SEFM15,Nik15} also studied another related problem
    of finding a $k$-dimensional simplex of maximum volume inside a given convex hull,
    which can be reduced to the former problem under cardinality constraint.
    For finding the principal $k\times k$ submatrix of
    a positive semidefinite matrix with the largest determinant,
    \cite{Nik15} gave an algorithm that
    obtains an 
    approximation of $e^{-(k + o(k))}$.
    On the hardness side,
    all these
    problems have been showed to be exponentially inapproximable~\cite{Kou06,CM13,SEFM15}.
    %Comparing to these works,
    %all matrices considered in our problem are symmetric, diagonally dominant,
    %M-matrices (SDDM), which is a subclass of PSDs.
    %This implies that
    %we can approximate the optimum solution faster using graph-specific algorithms,
    %and that the hardness of our problem does not follow from those for PSDs.

    %Among these works,
    %the problem studied in~\cite{CM09,CM13} can be described
    %as the following column selection problem:
    %Given a matrix $A_{d\times n}$ and an integer $k$,
    %select a subset $J \subseteq \setof{1,2,\dots,n}$ with $\sizeof{J} = k$
    %such that $\det(A_J^T A_J)$ is maximized.
    %Our problem, on the other hand, can be seen as a row selection problem:
    %since one can write a Laplacian as $ \LL = (\WW^{1/2} \BB)^T \WW^{1/2} \BB $,
    %the goal is to select $k$ rows of $\WW^{1/2}\BB$ so that the determinant is maximized.
    %Here $\BB_{m\times n}$ is the signed edge-vertex incidence matrix 
    %and $\WW_{m\times m}$ is the diagonal edge-weight matrix. 

    The problem studied in~\cite{Nik15} can also be stated as the following:
    \begin{quote}
        {\em 
        Given $m$ vectors $\xx_1,\ldots,\xx_m \in \mathbb{R}^n$ and
        an integer $k$, find a subset $S\subseteq [m]$ of cardinality $k$
        so that the product of the largest $l$ eigenvalues of the matrix
        $\SSi_S \defeq \sum\nolimits_{i\in S} \xx_i \xx_i^T$ is maximized
        where $l \defeq \min\{k, n\}$. $(\star)$
    }
    \end{quote}
    \cite{Nik15} gave a polynomial-time algorithm that
    obtains an $e^{-(k+o(k))}$-approximation
    when $k \leq n$.
    When $k \geq n$, Problem~$(\star)$ is equivalent to maximizing the determinant
    of $\SSi_S$ by selecting $k$ vectors.
    \cite{SX18} showed that one can obtain an $e^{-n}$-approximation
    for $k \geq n$. Moreover,
    they showed that given $k = \Omega(n/\eps + \log (1/\eps) / \eps^2)$,
    one can obtain a $(1 + \eps)^{-n}$-approximation.
    Using the algorithms in~\cite{Nik15,SX18}, we can obtain
    an $e^{-n}$-approximation to a problem
    of independent interest
    but different from Problem~\ref{prob:ecard}:
    Select at most $k$ edges from a candidate edge set
    to add to an empty graph so that the number of spanning trees
    is maximized.
    In contrast,
    in Problem~\ref{prob:ecard},
    we are seeking to add $k$ edges to a graph that is already connected.
    Thus, their algorithms cannot directly apply to Problem~\ref{prob:ecard}.
    %\cite{SX18} studies the problem of D-optimal experimental design.
    %In D-optimal experimental design,
    %one is given $n$ vectors $\xx_1, \ldots, \xx_n \in \mathbb{R}^p$,
    %and need to find a subset $S \subseteq [n]$ of cardinality at most $k$ ($k \geq p$)
    %so that the following quantity is minimized:
    %\begin{align*}
    %    \det\kh{ \sum\nolimits_{i\in S} \xx_i \xx_i^T }^{-1/p}.
    %\end{align*}
    %\cite{SX18} showed that one can obtain an $e$-approximation
    %in polynomial time.

    In~\cite{ALSW17a,ALSW17b},
    the authors also studied Problem~$(\star)$. % the problem of D-optimal experimental design.
    They gave an algorithm that, when $k = \Omega(n / \eps^2)$,
    gives a $(1 + \eps)^{-n}$-approximation.
    Their algorithm first computes a fractional solution
    using convex optimization
    and then rounds the fractional solution to integers
    using spectral sparsification.
    Since spectral approximation is preserved
    under edge additions,
    their algorithm can apply to Problem~\ref{prob:ecard}
    obtaining a $(1 + \eps)^{-n}$-approximation.
    %Using their algorithm,
    %we can obtain a $(1 + \eps)^{-n}$-appromxation to
    %Problem~\ref{prob:ecard}.
    However, their algorithm
    needs $k$ to be $\Omega(n/\eps^2)$, given
    that the candidate edge set
    is supported on $O(n)$ vertices
    (which is natural in real-world datasets~\cite{KGS+11,KHD16}).
    In contrast,
    $k$ could be arbitrarily smaller than $n$ in our setting.
    %(b) obtains an approximation to the determinant itself
    %instead of its increase, which is in contrast to our algorithm.
    %An approximation to the increase in the determinant
    %will give a better approximation
    %to the optimum solution
    %when the increase is not comparable to the original determinant
    %(probably due to a small $k$)
    %which is typical in applications~\cite{KSHD16a,KSHD16b}.
    %In addition, our approximation ratio is at least as good
    %as $c^{-n}$ for some $c > 0$ when the optimum increase
    %in the determinant is at most $e^{O(n)}$,
    %which is true for real-world networks that are usually sparse.

    We remark that both our setting
    of adding edges to a connected graph
    and the scenario that $k$ could be arbitrarily smaller than $n$
    have been used
    in previous works solving graph optimization problems such as maximizing the
    algebraic connectivity of a graph~\cite{KMST10,NXCGH10,GB06}.
    We also remark that the algorithms in~\cite{Nik15,SX18,ALSW17a,ALSW17b}
    all need to solve a convex optimization for their continuous relaxation,
    which runs in polynomial time in contrast to our nearly-linear running time,
    while the efficiency is crucial in applications~\cite{KHD16,KSHD16a,SCL16,SM18}.
    %One could use the algorithm of~\cite{Nik15} to
    %give an $e^{-(k+o(k))}$-approximation to Problem~\ref{prob:ecard}.
    %This approximation could be either better or worse than ours
    %depending on the optimum solution.
    %Similar to~\cite{ALSW17a,ALSW17b}, this algorithm
    %needs to solve a convex optimization to obtain an optimum
    %fractional solution, which is a bottleneck for its efficiency.

    We also note that~\cite{DPPR17} gave an algorithm that computes the determinant
    of an SDDM matrix to a $(1\pm\eps)$-error in $\Otil(n^2 \eps^{-2})$ time.
    In our algorithm, we are able to maximize the determinant
    in nearly-linear time without computing it.
    %As for hardness,
    %if we change the cardinality constraint of Problem~\ref{prob:ecard}
    %to a partition constraint, the problem will contain
    %Nash social welfare maximization problem as a special case,
    %and thus its exponential inapproximability will follow
    %from the hardness results in~\cite{Lee17}.
    %However,
    %we are not aware of any existing proof for the inapproximaility
%of Problem~\ref{prob:ecard} with a cardinality constraint.
    %\hl{One of the reviewers pointed out [Lee17] and the reduction when adding a
    %    partition constraint. Note that cardinality constraints are subsets of
    %    partition constraints.
    %So this reduction cannot prove hardness of Problem~\ref{prob:ecard}.}
}

\paragraph{Fast Computation of Effective Resistances.}{
    Fast computation of effective resistances
    has various applications in sparsification~\cite{SS11,ADK+16,LS17},
    sampling random spanning trees~\cite{DKP+17,DPPR17,Sch18},
    and solving linear systems~\cite{KLP16}.
    \cite{SS11,KLP16} gave approximation routines that,
    using Fast Laplacian Solvers~\cite{ST14,CKMPPRX14},
    compute effective resistances for all edges to $(1\pm \eps)$-errors
    in $\Otil(m\eps^{-2})$ time.
    \cite{CGP+18} presents an algorithm that computes the effective resistances
    of all edges
    to $(1 \pm \eps)$-errors in $O(m^{1 + o(1)} \eps^{-1.5})$ time.
    For computing effective resistances for a given set
    of vertex pairs,
    \cite{DKP+17} gave a routine that,
    using divide-and-conquer based on Schur complements approximation~\cite{KS16},
    computes the effective resistances between $\cen$ pairs of
    vertices to $(1\pm \eps)$-errors in $\Otil(m + (n + \cen)\eps^{-2})$ time.
    \cite{DKP+17} also used a divide-and-conquer to sample random spanning trees
    in dense graphs faster.
    For maintaining $(1 + \eps)$-approximations to all-pair effective resistances
    of a fully-dynamic graph,
    \cite{DGGP18} gave a data-structure with
    $\Otil(m^{4/5} \eps^{-4})$ expected amortized update and query time.
    In~\cite{DPPR17}, the authors combined the divide-and-conquer idea
    and their determinant-preserving sparsification to further accelerate
    random spanning tree sampling in dense graphs.
    A subset of the authors of this paper (Li and Zhang)
    recently~\cite{LZ18}
    used a divide-and-conquer approach to compute,
    for every edge $e$,
    the sum of effective resistances between all vertex pairs
    in the graph in which $e$ is deleted.
    Our routine for performing fast sequential updates in Section~\ref{sec:seq}
    is motivated by these divide-and-conquer methods
    and is able to cope with
    an online-offline hybrid setting.
    %interpolation \hl{Is this word used properly?}
    %between online and offline.
}

\paragraph{Top-$k$ Selections on Graphs.}{
    Conventional top-$k$ selections on graphs
    that rely on submodularity
    usually run in $\Omega(mk)$ time, where $m$ is the number of edges.
    Here, we give a few examples of them.

    In~\cite{BBCL14}, the authors studied the problem of maximizing
    the spread of influence through a social network.
    Specifically, they studied the problem of finding a set of
    $k$ initial seed vertices in a network so that,
    under the independent cascade model~\cite{KKT03} of network diffusion,
    the expected
    number of vertices reachable from the seeds is maximized.
    Using hypergraph sampling, the authors gave a greedy algorithm
    that achieves a $(1 - \frac{1}{e} - \eps)$-approximation in $O((m + n)k\eps^{-2}\log n)$
    time.

    \cite{Yos14,MTU16} studied the problem of finding
    a vertex set $S$ with maximum betweenness centrality
    subject to the constraint $\sizeof{S}\leq k$.
    %The betweenness centrality of a vertex set $S$ is defined as
    %\[
    %    \sum_{s,t\in V} \frac{\sigma_{s,t}(S)}{\sigma_{s,t}},
    %\]
    %where $\sigma_{s,t}$ is the number of $s$-$t$ shortest paths,
    %and $\sigma_{s,t}(S)$ is the number of $s$-$t$ shortest paths
    %that have at least one internal vertex in $S$.
    Both algorithms in~\cite{Yos14,MTU16}
    are based on sampling shortest paths.
    To obtain a $(1 - \frac{1}{e} - \eps)$-approximation,
    their algorithms need at least $\Omega(mn\eps^{-2})$ running time
    according to Theorem 2 of~\cite{MTU16}.
    Given the assumption that
    the maximum betweenness centrality among all sets of $k$ vertices
    is $\Theta(n^2)$, the algorithm in~\cite{MTU16}
    is able to obtain a solution with the same approximation ratio in $O((m+n)k\eps^{-2}\log n)$ time.

    A subset of the authors of this paper (Li, Yi, and Zhang)
    and Peng and Shan
    recently~\cite{LPS+18}
    studied the problem of finding a set $S$ of $k$
    vertices so that
    the quantity
    \[
        \sum_{u\in (V\setminus S)} \er(u,S)
    \]
    is minimized.
    %where $\er(u,S)$ denotes the effective resistance between vertex $u$
    %and vertex set $S$.
    Here $\er(u,S)$ equals,
    in the graph in which $S$ is identified as a new vertex,
    the effective resistance between $u$ and the new vertex.
    By computing marginal gains for all vertices
    in a way similar to the effective resistance estimation routine
    in~\cite{SS11},
    the authors achieved a $(1 - \frac{k}{k-1}\cdot \frac{1}{e} - \eps)$-approximation
    in $\Otil(mk\eps^{-2})$
    time.
    
    %\todo{add one more example}
    %In~\cite{KSG08},

    We remark that there are algorithms for maximizing submodular functions
    that use only nearly-linear evaluations of the objective function~\cite{BV14,EN17}.
    However,
    in many practical scenarios,
    evaluating the objective function
    or the marginal gain is expensive.
    Thus, directly applying those algorithms usually
    requires superlinear or even quadratic running time.
}

\section*{Acknowledgements}

We thank Richard Peng, He Sun, and Chandra Chekuri for very helpful discussions. %and comments on our manuscript.

%\subsection{Organization}

\section{Preliminaries}

%\subsection{Notations}
%
%We use normal lowercase letters like $a,b,c$ to denote scalars in $\mathbb{R}$,
%normal uppercase letters like $A, B, C$ to denote sets, bold lowercase letters like $\aa,\bb,\cc$ to denote vectors,
%and bold uppercase letters like $\AA, \BB, \CC$ to denote matrices.
%We write $\aa_{[i]}$ to denote the $i^{\mathrm{th}}$ entry of vector $\aa$
%and $\AA_{[i,j]}$ to denote entry $(i,j)$ of matrix $\AA$.
%We also write $\AA_{[i,:]}$ to denote the $i^{\mathrm{th}}$ row of $\AA$
%and $\AA_{[:,j]}$ to denote the $j^{\mathrm{th}}$ column of $\AA$.
%
%We write sets in matrix subscripts to denote submatrices. For example,
%$\AA_{[I,J]}$ denotes the submatrix of $\AA$ with row indices in $I$
%and column indices in $J$.
%To simplify notation,
%we also write $\AA_{-i}$ to denote the submatrix of $\AA$
%obtained by removing the $i^{\mathrm{th}}$ row and $i^{\mathrm{th}}$ column of $\AA$.
%For example, for an $n\times n$ matrix $\AA$,
%$\AA_{-n}$ denotes the submatrix $\AA_{[1:n-1,1:n-1]}$.
%
%Note that the precedence of matrix subscripts is
%the lowest. Thus, $\AA_{-n}^{-1}$ denotes the inverse of $\AA_{-n}$
%instead of a submatrix of $\AA^{-1}$.
%
%For two matrices $\AA$ and $\BB$, we write $\AA \preceq \BB$ to denote
%that $\BB - \AA$ is positive semidefinite, i.e.,
%$\xx^T \AA \xx \leq \xx^T \BB \xx$ holds for every real vector $\xx$.
%
%We use $\ee_i$ to denote the $i^{\mathrm{th}}$ standard basis vector of appropriate dimension,
%and $\one_{S}$ to denote the indicator vector of $S$.

\subsection{Graphs, Laplacians, and Effective Resistances}

Let $G = (V, E)$ be a positively weighted undirected
graph. $V$ and $E$ is respectively the vertex set and the edge set of 
the graph, and $w: E\to \mathbb{R}^+$ is the weight 
function. Let $|V|=n$ and $|E|=m$. The Laplacian matrix $\LL$ of $G$ is given by
\begin{align*}
	\LL^G_{[u,v]} =
	\begin{cases}
		-w(u,v) & \text{if $u\sim v$}, \\
		\mathrm{deg}(u) & \text{if $u=v$}, \\
		0 & \text{otherwise},
	\end{cases}
\end{align*}
where $\mathrm{deg}(u) \defeq \sum\nolimits_{u\sim v} w(u,v)$,
and we write $u\sim v$ iff $(u,v) \in E$. We will use $\LL$ and $\LL^{G}$
interchangeably when the context is clear.
%Let $w_{\max}$ be the maximum weight
%and $w_{\min}$ be the minimum weight among all edges.

If we assign an arbitrary orientation to each edge of $G$,
we obtain a signed edge-vertex incident matrix $\BB_{m\times n}$ of graph $G$
defined as
\[
    \BB_{[e,u]} = \begin{cases}
        1 & \text{if $u$ is $e$'s head}, \\
        -1 & \text{if $u$ is $e$'s tail}, \\
        0 & \text{otherwise}.
    \end{cases}
\]
Let $\WW$ be an $m\times m$ diagonal matrix in which $\WW_{[e,e]} = w(e)$.
Then we can express $\LL$ as $\LL=\BB^T\WW\BB$.
It follows that a quadratic form of $\LL$ can be written as
\[
    \xx^T \LL \xx = \sum\nolimits_{u\sim v} w(u,v) (\xx_{[u]} - \xx_{[v]})^2.
\]
%\todo{write $\LL = \sum_{e\in E} w(e) \bb_e \bb_e^T$}
It is then observed that $\LL$ is positive semidefinite,
and $\LL$ only has one zero eigenvalue if $G$ is a connected graph.  
If we let $\bb_e$ be the $e^{\mathrm{th}}$ column of $\BB^T$,
we can then write $\LL$ in a sum of rank-1 matrices as
$\LL = \sum\nolimits_{e\in E} w(e) \bb_e \bb_e^T$.

The Laplacian matrix is related to the number of spanning
trees $\Tcal(G)$ by Kirchhoff's matrix-tree
theorem~\cite{Kir47}, which expresses $\Tcal(G)$ using
any $(n-1)\times (n-1)$ principle minors of $\LL$.
We denote by $\LL_{-u}$ the principle submatrix derived
from $\LL$ by removing the row and column
corresponding to vertex $u$. Since the removal of any
matrix leads to the same result, we will usually
remove the vertex with index $n$. Thus, we write the
Kirchhoff's matrix-tree theorem as
\begin{align}
\Tcal(G) = \det(\LL_{-n}^G)\,.
\end{align}

The effective resistance between any pair of vertices can be defined by 
a quadratic form of the Moore-Penrose inverse  $\LL^\dag$ of the Laplacian matrix~\cite{KR93}.
\begin{definition}%[Effective Resistance Between Two Vertices~\cite{KR93}]
	%Given a connected undirected graph $G = (V,E)$ with $n$ vertices, $m$ edges, and positive edge weights
	%$w : E \to \rea_{+}$.
	Given a connected graph $G = (V,E,w)$ with Laplacian matrix $\LL$,
	the effective resistance any two vertices $u$ and $v$ is defined as
	$
		\er(u,v) = \kh{\ee_u - \ee_v}^T \LL^\dag \kh{\ee_u - \ee_v}\,.
		%= \LL^\dag(u,u) - \LL^\dag(u,v) - \LL^\dag(v,u) + \LL^\dag(v,v).
	$
%	where $\ee_i$ denotes the $i^{\mathrm{th}}$ standard basis vector.
\end{definition}

%\iffalse
%The effective resistance can also be defined between a vertex and a vertex set.
%
%\begin{definition}[Effective Resistance Between a Vertex and a Vertex Set~\cite{CP11}]
%	For a connected graph $G = (V,E,w)$ with Laplacian matrix $\LL$,
%	the effective resistance between vertices $u$ and $S$ is defined as
%	\begin{align*}
%		\er^G(u,S) = \kh{\LL^{-1}_{-S}}_{[u,u]}.
%	\end{align*}
%\end{definition}
%
%\begin{remark}\label{rem:eer}
%$\er(u,S)$ can be interpreted as the electrical
%resistance when the graph is treated as a resistor network
%in which vertices in $S$ are grounded.
%Namely, if we treat every edge $e$
%as a resistor with resistance $r_e = 1 / w(e)$
%and ground all vertices in $S$,
%then $\er(u,S)$ equals the voltage of $u$
%when an electrical flow sends one unit of current into $u$ and
%removes one unit of current from $S$ (i.e. the ground).
%\end{remark}
%
%\fi

%\begin{definition}[Effective Resistance Between a Vertex and a Vertex Set]\label{def:ervs}
%	For a connected graph $G = (V,E,w)$ with Laplacian matrix $\LL$,
%	the effective resistance between a vertex $u$ and a vertex set $S$ such that $u\notin S$ is defined as
%	\begin{align*}
%		\er^G(u,S) = \kh{\LL_{-S}}^{-1}(u,u).
%	\end{align*}
%\end{definition}

For two matrices $\AA$ and $\BB$, we write
$\AA \pleq \BB$ to denote $\xx^T \AA \xx \leq \xx^T \BB \xx$
for all vectors $\xx$.
If for two connected graph $G$ and $H$ their Laplacians
satisfy $\LL^G \pleq \LL^H$,
then $\kh{\LL^H}^\dag \pleq \kh{\LL^G}^\dag$.

\subsection{Submodular Functions}

We next give the definitions for monotone and
submodular set functions.
For conciseness we use $S + u$ to denote $S\union\setof{u}$.
%and $S - u$ to denote $S\setminus \setof{u}$.

\begin{definition}[Monotonicity]
A set function $f : 2^V \to \mathbb{R}$ is
monotone if
$
	f(S)\leq f(T)
$
holds for all $S\subseteq T \subseteq V$.
\end{definition}

\begin{definition}[Submodularity]
A set function $f : 2^V \to \mathbb{R}$ is
submodular if
$
	f(S + u) - f(S)  \geq f(T + u) - f(T)
$
holds for
all $S\subseteq T \subseteq V$ and $u\in V\setminus T$.
\end{definition}

\subsection{Schur Complements}

Let $V_1$ and $V_2$ be a partition of vertex set $V$,
which means $V_2= V\setminus V_1$.
Then, we decompose the Laplacian into
blocks using $V_1$ and $V_2$ as the block indices:
\[
\LL =
\begin{pmatrix}
\LL_{[V_1, V_1]} & \LL_{[V_1, V_2]}\\
\LL_{[V_2, V_1]} & \LL_{[V_2, V_2]}
\end{pmatrix}.
\]
The Schur complement of $G$, or $\LL$, onto $V_1$ is defined as:
\[
SC(G,V_1)
= SC(\LL^{G}, V_1)
\defeq  \LL^{G}_{[V_1, V_1]}
-  \LL^{G}_{[V_1, V_2]} \left( \LL^{G}_{[V_2, V_2]} \right)^{-1}
	\LL^{G}_{[V_2, V_1]},
\]
and we will use $SC(G,V_1)$ and $SC(\LL^{G},V_1)$
interchangeably.
%We further note that when the context
%is clear, we always consider $V_1$ to be the vertex 
%set we Schur complement onto, and $V_2$ to be the vertex set we eliminate.

%Schur complements behave nicely with respect to determinants determinants, which suggests the general structure of the recursion we will use for estimating the determinant.
%\begin{fact}\label{fact:detMinor}
%For any matrix $\MM$ where $\MM_{[V_2, V_2]}$ is invertible,
%\[
%\det{(\MM_{-n})}
%= \det{\left(\MM_{\left[V_2, V_2\right]}\right)}
%\cdot \detp{\left(\sc{\MM}{V_1}\right)}.
%\]
%\end{fact}
%This relationship also suggests that
%there should exist a bijection between spanning tree distribution in $G$ and the product distribution given by sampling
%spanning trees independently from
%$\sc{\LL}{V_1}$ and the graph Laplacian formed by adding one
%row/column to $\LL_{[V_2, V_2]}$.

%Finally, our algorithms for approximating Schur
%complements rely on the fact that they preserve certain marginal probabilities. The algorithms of ~\cite{ColbournDN89,ColbournMN96,HarveyX16,DurfeeKPRS16} also use variants of some of these facts, which are closely related to the preservation of the spanning tree distribution on $\sc{\LL}{V_1}$. (See 
%Section~\ref{sec:spanning_tree} for details.)
The Schur complement preserves the effective resistance between vertices $u,v\in V_1$.
\begin{fact}
	\label{fact:SchurResistance}
	Let $V_1$ be a subset of vertices of a graph $G$. Then for any vertices
	$u, v \in V_1$, we have:
	\[
	\er^{G}\left(u, v\right)
	= \er^{SC(G,V_1)}\left(u, v\right).
	\]
\end{fact}

\section{Nearly-Linear Time Approximation Algorithm}

%We use notation $-e$ to denote deletion of edge $e$ and $+e$ to denote addition of edge $e$.
By the matrix determinant lemma~\cite{Har97},
we have
\begin{align*}
	\det\kh{\kh{\LL + w(u,v)\bb_{u,v}\bb_{u,v}^T}_{-n}} 
	=
	\kh{1 + w(u,v)\bb_{u,v}^T \LL^\dag \bb_{u,v}}
	\det\kh{\LL_{-n}}.
	%\er(u,v) \defeq \cchi_{u,v}^T \LL^\dag \cchi_{u,v} = \begin{cases}
	%	\frac{\Tcal(G) - \Tcal(G-(u,v))}{\Tcal(G)}, & (u,v) \in E \\
	%	\frac{\Tcal(G + (u,v)) - \Tcal(G)}{\Tcal(G)}, & (u,v) \in E
	%\end{cases}.
\end{align*}
Thus, by Kirchhoff's matrix tree theorem,
we can write the increase of $\log \Tcal(G)$ upon the addition of edge $(u,v)$ as
\begin{align*}
	\log \Tcal\kh{G + (u,v)} - \log \Tcal(G) =
    \log \kh{1 + w(u,v)\er^G(u,v)},
\end{align*}
which immediately implies $\log \Tcal(G)$'s submodularity
by Rayleigh's monotonicity law~\cite{Sto87}.

\begin{lemma}
    $\log \Tcal(G + \ses)$ is a monotone submodular function.
\end{lemma}

Thus, one can obtain a $(1 - \frac{1}{e})$-approximation for Problem~\ref{prob:ecard} by
a simple greedy algorithm that, in each of $k$ iterations, selects 
 the edge that results in the largest   effective resistance times edge weight~\cite{NWF78}.

Our algorithm is based on another greedy algorithm for maximizing
a submodular function, which is proposed in~\cite{BV14}.
Instead of selecting the edge with highest effective resistance
in each iteration,
the algorithm maintains a geometrically decreasing threshold and
sequentially selects any edge with effective resistance above
the threshold.
The idea behind this greedy algorithm is that one can
always pick an edge with highest effective resistance up to a $(1\pm\eps)$-error.
In doing so, the algorithm is able to obtain a $(1-\frac{1}{e}-\eps)$-approximation
using nearly-linear marginal value evaluations.
We give this algorithm in Algorithm~\ref{alg:GreedyTh}.

\begin{algorithm}
    \caption{$P = \GreedyTh(G, \ces, w, \eps, k)$}
	\label{alg:GreedyTh}
	\Input{
        $G = (V,E)$: A connected graph. \\
        $\ces$: A candidate edge set with $\sizeof{\ces} = \cen$. \\
        $w : (E\union \ces) \to \mathbb{R}^+$: An edge weight function. \\
        $\eps$: An error parameter. \\
        $k$: Number of edges to add.
	}
	\Output{
        $\ses$ : A subset of $\ces$ with at most $k$ edges.
	}
    $\ses \gets \emptyset$\;
    $er_{max} \gets \max\nolimits_{(u,v)\in \ces} w(u,v)\cdot \er^G(u,v)$ \label{line:ermax}\;
    $th \gets er_{max}$\;
    \While{$th \geq \frac{\eps}{\cen} er_{max}$}{
         \ForAll{$(u,v) \in \ces\setminus \ses$}{\label{line:for1}
             \If{$\sizeof{\ses} < k$ {\rm and $w(u,v)\cdot \er^G(u,v) \geq th$}}{
                $G\gets G + (u,v)$\;
                $\ses \gets \ses \union \setof{(u,v)}$
                \label{line:for2}
            }
        }
        $th \gets (1 - \eps) th$
    }
    \KwRet{\ses}
\end{algorithm}

The performance of algorithm $\GreedyTh$ is characterized in the following theorem.
%according to~\cite{BV14}.

\begin{theorem}%[\cite{BV14}]
    The algorithm $\ses = \GreedyTh(G,\ces,w,\eps,k)$
    takes a connected graph ${G = (V,E)}$ with $n$ vertices and $m$ edges,
    %such that $\Tcal(G) \geq 1$,
    an edge set $\ces$ of $\cen$ edges, %such that $\ces\intersect E = \emptyset$,
    an edge weight function ${w : (E \union \ces) \to \mathbb{R}^+}$,
    a real number $0 < \eps \leq 1/2$,
    and an integer $1\leq k\leq \cen$,
    and returns an edge set $\ses \subseteq \ces$ of at most $k$ edges.
    The algorithm computes effective resistances for
    $O(\frac{\cen}{\eps} \log \frac{\cen}{\eps})$ pairs of vertices,
    and uses $O(\frac{\cen}{\eps} \log \frac{\cen}{\eps})$ arithmetic operations.
    %effective resistance computations and arithmetic operations.
    %by computing effective resistances for $O(\frac{\cen}{\eps} \log \frac{\cen}{\eps})$
    %pairs of vertices.
    The $\ses$ returned satisfies the following statement:
    \begin{align*}
        \log \frac{\Tcal(G + \ses)}{\Tcal(G)} \geq
        \kh{1 - \frac{1}{e} - \eps} \log \frac{\Tcal(G + O)}{\Tcal(G)},
    \end{align*}
    where $O \defeq \argmax\nolimits_{S \subseteq \ces, \sizeof{S}\leq k} \Tcal(G + S)$
    denotes an optimum solution.
\end{theorem}

A natural idea to accelerate the algorithm $\GreedyTh$
is to compute effective resistances approximately, instead
of exactly, using the routines in~\cite{SS11,DKP+17}.
To this end, we develop the following lemma,
which shows that
to obtain a multiplicative approximation of
\[
    \log \kh{1 + w(u,v)\er^G(u,v)},
\]
it suffices to compute a multiplicative approximation of $\er^G(u,v)$. We note that if given that $a$ and $b$ are within a factor of $(1+\epsilon)$ of each other, then it follows that $\log{(1+a)}$ and $\log{(1+b)}$ are within $(1+\epsilon)$ as well, as the function $\log\log(1+e^x)$ is a $1$-Lipschitz function. Since we are using $(1\pm\eps)$-approximation, we give an alternative proof.

\begin{restatable}[]{lemma}{logapprox} \label{lem:logapprox}
	For any non-negative scalars $a,b$ and $0 < \eps \leq 1/2$ such that
	\[
		(1 - \eps) a \leq b \leq (1 + \eps) a,
	\]
	the following statement holds:
	\[
		(1 - 2\eps) \log(1 + a) \leq \log(1 + b) \leq (1 + 2\eps) \log(1 + a).
	\]
\end{restatable}
\begin{proof}
	Since $\log (1 + a) = \ln (1 + a) / \ln 2$ and $\log (1 + b) = \ln (1 + b) / \ln 2$,
	we only need to prove the $(1 + 2\eps)$-approximation between
	$\ln(1+a)$ and $\ln(1+b)$.
	
	Let $x_2 \approx 2.51$ be the positive root of the equation
	\[
		\ln (1 + x) = x / 2.
	\]
	Then, we have $\ln(1 + x) \geq x/2$ for $0 \leq x \leq x_2$ and
	 $\ln(1 + x) < x/2$ for $x>x_2$.
	 %reverse elsewhere 
	 %\sep{Should this be  $\ln(1 + x) < x/2$ for $x>x_2$?}.
	
	For $a \leq x_2$, we have
	\begin{align*}
		& \frac{\abs{\ln(1 + a) - \ln(1 + b)}}{\abs{a - b}} \leq
			\left.\frac{\mathrm{d} \ln(1 + x)}{\mathrm{d} x}\right|_{x = \min\{a,b\}}
			\leq \left.\frac{\mathrm{d} \ln(1 + x)}{\mathrm{d} x}\right|_{x = 0}
			= 1 \\
		\Rightarrow \quad &
		\abs{\ln(1 + a) - \ln(1 + b)} \leq \abs{a-b} \leq \eps a \leq 2\eps \ln(1 + a)
			\qquad \text{by $\ln(1 + a) \geq a/2$} \\
		\Rightarrow \quad &
		(1 - 2\eps) \ln(1 + a) \leq \ln(1 + b) \leq (1 + 2\eps) \ln(1 + a).
	\end{align*}
	
	For $a > x_2$, we have
	\begin{align*}
		& \frac{\abs{\ln(1 + a) - \ln(1 + b)}}{\abs{a - b}} \leq
			\left.\frac{\mathrm{d} \ln(1 + x)}{\mathrm{d} x}\right|_{x = \min\{a,b\}}
			\leq \frac{1}{1 + (1 - \eps)a}
			\\
		\Rightarrow \quad &
		\abs{\ln(1 + a) - \ln(1 + b)} \leq \frac{\abs{a - b}}{1 + (1 - \eps)a}
		\leq \frac{\eps a}{(1 - \eps)a} \leq \frac{\eps a}{\frac{1}{2} a}
		\qquad \text{by $0 < \eps \leq 1/2$} \\
 		= & 2\eps \leq 2\eps \ln(1 + a) \qquad \text{by $a > x_2\approx 2.51$} \\
		\Rightarrow \quad &
		(1 - 2\eps) \ln(1 + a) \leq \ln(1 + b) \leq (1 + 2\eps) \ln(1 + a).
	\end{align*}
\end{proof}

By using the effective resistance approximation routine in~\cite{DKP+17},
one can pick an edge with effective resistance above the threshold up to a $1\pm \eps$ error.
Therefore,
by an analysis similar to that of Algorithm~1 of~\cite{BV14},
one can obtain a $(1 - \frac{1}{e} - \eps)$-approximation in time
$\Otil((m + (n + \cen)\eps^{-2})k)$.
The reason that the running time has a factor $k$ is that one has to
recompute the effective resistances whenever an edge is added to the graph.
To make the running time independent of $k$, we will need a faster algorithm
for performing the sequential updates, i.e., Lines~\ref{line:for1}-\ref{line:for2}
of Algorithm~\ref{alg:GreedyTh}.

\subsection{Routine for Faster Sequential Edge Additions}\label{sec:seq}

We now use the idea we stated in Section~\ref{sec:ideas}
to perform the sequential updates at Lines~\ref{line:for1}-\ref{line:for2}
of Algorithm~\ref{alg:GreedyTh} in nearly-linear time.
We use a routine from~\cite{DKP+17} to compute the approximate
Schur complement:

\begin{lemma}\label{lem:ApproxSchur}
    There is a routine
    $\SS = \ApproxSchur(\LL, C, \eps, \delta)$ that takes
    a Laplacian $\LL$ corresponding to graph $G = (V,E)$ with $n$ vertices and $m$ edges,
    a vertex set $C\subseteq V$, and real numbers $0 < \eps \leq 1/2$ and $0 < \delta < 1$,
    and returns a graph Laplacian $\SS$ with $O(\sizeof{C}\eps^{-2}\log n)$ nonzero entries
    supported on $C$.
    With probability at least $1 - \delta$, $\SS$ satisfies
    \begin{align*}
        (1 - \eps) SC(\LL, C) \pleq \SS \pleq (1 + \eps) SC(\LL, C).
    \end{align*}
    The routine runs in $\Otil(m\log^2(n/\delta) + n\eps^{-2}\log^4(n/\delta))$ time.
\end{lemma}

We give the routine for performing fast sequential updates
in Algorithm~\ref{alg:AddAbove}.
%We also illustrate the main steps of this routine in Figure~\todo{}.

\quad

\begin{algorithm}[H]
    \caption{$P = \AddAbove(G, (u_i,v_i)_{i=1}^{\cen}, w, th, \eps, k)$}
	\label{alg:AddAbove}
	\Input{
        $G = (V,E)$: A connected graph. \\
        $(u_i,v_i)_{i=1}^{\cen}$: An edge sequence of $\cen$ edges. \\
        $w : (E\union (u_i,v_i)_{i=1}^{\cen}) \to \mathbb{R}^+$: An edge weight function. \\
        $th$: A threshold. \\
        $\eps$: An error parameter. \\
        $k$: Number of edges to add.
	}
	\Output{
        $\ses$ : A subset of $\ces$ with at most $k$ edges.
	}
    Let $\eps_1 = \frac{2}{3}\cdot\eps / \log \cen$ and
    $\eps_2 = (1 - 1/\log\cen)\cdot\eps$.\label{line:eps}\;
    Let $\LL$ be the Laplacian matrix of $G$.\;
    \If{$\cen = 1$}{
        $\SS \gets\ApproxSchur(\LL,\setof{u_1,v_1},\eps, \frac{1}{10n(m + q)})$ \;
        \label{line:1}
        Compute $\SS^\dag$ by inverting $\SS$ in $O(1)$ time.\;
        %If $w(u_1,v_1)\cdot \bb_{u_1,v_1}^T \LL^\dag \bb_{u_1,v_1} \geq th$,
        %return $\setof{(u_1,v_1)}$; otherwise return $\emptyset$.
        \If{$w(u_1,v_1)\cdot \bb_{u_1,v_1}^T \SS^\dag \bb_{u_1,v_1} \geq th$\hspace{3pt} {\rm and}\hspace{3pt}$k > 0$ }{
            \KwRet $\setof{(u_1,v_1)}$
        }\Else{
            \KwRet $\emptyset$ \label{line:10}
        }
    }
    \Else{
        Divide $(u_i,v_i)_{i=1}^{\cen}$ into two intervals \label{line:>1}
        \vspace{-5pt}
        \[
            \vspace{-3pt}
            f^{(1)} \defeq (u_1,v_1),\ldots, (u_{\floor{\cen / 2}},v_{\floor{\cen / 2}})
            \quad
            \text{and}
            \quad
            f^{(2)} \defeq (u_{\floor{\cen / 2} + 1},v_{\floor{\cen / 2} + 1}),\ldots,(u_{\cen},v_{\cen}),
        \]
        %$f^{(1)} \defeq (u_i,v_i)_{i=1}^{\floor{\cen / 2}}$
        %and
        %$f^{(2)} \defeq (u_i,v_i)_{\floor{\cen / 2} + 1}^{\cen}$,
        and let $V^{(1)}$ and $V^{(2)}$ be the respective set of
        endpoints of edges in $f^{(1)}$ and $f^{(2)}$.\;
        $\SS^{(1)} \gets\ApproxSchur(\LL,V^{(1)},\eps_1, \frac{1}{10n(m + q)})$\;
        $P^{(1)} \gets \AddAbove(\SS^{(1)}, f^{(1)}, w, th, \eps_2, k)$\;
        Update the graph Laplacian by:
        \vspace{-5pt}
        \begin{align*}
        %$
        \LL \gets \LL + \sum\nolimits_{(u,v)\in P^{(1)}} w(u,v) \bb_{u,v}\bb_{u,v}^T.
        %$
        \end{align*}
        \vspace{-20pt}
        \;
        $\SS^{(2)} \gets\ApproxSchur(\LL,V^{(2)},\eps_1, \frac{1}{10n(m + q)})$\;
        $P^{(2)} \gets \AddAbove(\SS^{(2)}, f^{(2)}, w, th, \eps_2, k - \sizeof{P^{(1)}})$\;
        \KwRet{$P^{(1)} \union P^{(2)}$} \label{line:>10}
    }
\end{algorithm}

\quad

The performance of $\AddAbove$ is characterized in Lemma~\ref{lem:ds}.

%\lemalgo*

\begin{proof}[Proof of Lemma~\ref{lem:ds}]
    We first prove the correctness of this lemma by induction on $\cen$.

    When $\cen = 1$, the routine goes to lines~\ref{line:1}-\ref{line:10}.
    Lemma~\ref{lem:ApproxSchur} guarantees that $\SS$ satisfies
    \[
        (1 - \eps) SC(\LL,\setof{u_1,v_1}) \pleq \SS \pleq
        (1 + \eps) SC(\LL,\setof{u_1,v_1}),
    \]
    which implies
    \[
        (1 - 2\eps) \bb_{u_1,v_1}^T \LL^\dag \bb_{u_1,v_1} \leq \bb_{u_1,v_1}^T \SS^\dag \bb_{u_1,v_1} \leq
        (1 + 2\eps) \bb_{u_1,v_1}^T \LL^\dag \bb_{u_1,v_1}.
    \]
    Thus, the correctness holds for $\cen = 1$.

    Suppose the correctness holds for all $1 \leq \cen \leq t$ where $t\geq 1$.
    We now prove that it also holds for $\cen = t + 1$.
   Since $\cen > 1$, the routine goes to Line~\ref{line:>1}-\ref{line:>10}.
    Again by Lemma~\ref{lem:ApproxSchur}, we have
    \begin{align}\label{eq:s1}
        (1 - \frac{2}{3}\cdot\eps / \log \cen) SC(\LL,V^{(1)}) \pleq \SS^{(1)} \pleq
        (1 + \frac{2}{3}\cdot\eps / \log \cen) SC(\LL,V^{(1)}).
    \end{align}
    By the inductive hypothesis, any effective resistance query in $f^{(1)}$ is answered
    with an error within
    \[
        1 \pm (1 - 1/\log\cen)\cdot\eps.
    \]
    Combining this with~(\ref{eq:s1}) gives the correctness
    for $f^{(1)}$. Then, by a similar analysis, we can obtain the correctness
    for $f^{(2)}$. By induction, the correctness holds for all $\cen$.

    For the success probability,
    note that every time we invoke the routine $\ApproxSchur$,
    we set the failure probability to $\frac{1}{10n(m+q)}$.
    Thus, we get high probability by a union bound.

    We next analyzie the running time.

    Let $T(\cen, \eps)$ denote the running time of
    $\AddAbove(G, (u_i,v_i)_{i=1}^{\cen}, w, th, \eps, k)$
    when the number of edges in $G$ is $O(\cen\eps^{-2}\log n)$.
    It immediately follows that when $G$ contains $m$ edges, where $m$ is an arbitrary number,
    the total running time of $\AddAbove$ is at most
    \begin{align}
        2\cdot T(\cen/2, (1 - 1 / \log \cen)\cdot \eps) +
        \Otil(m + n\eps^{-2}\log^{2}\cen),
        \label{eq:totalrt}
    \end{align}
    since by Lemma~\ref{lem:ApproxSchur}, in the first step of divide-and-conquer,
    the routine will divide the graph into two Schur complements each with $O(\cen\eps^{-2}\log n)$ edges.
    Also by Lemma~\ref{lem:ApproxSchur}, we can write $T(\cen,\eps)$ in the following
    recurrence form:
    \begin{align}
        T(\cen,\eps) = 2\cdot T(\cen / 2, (1 - 1 / \log \cen)\cdot \eps) +
        \Otil(\cen\eps^{-2}\log^2 \cen),
        \label{eq:recur}
    \end{align}
    which gives $T(\cen,\eps) = \Otil(\cen\eps^{-2})$.
    Combining this with~(\ref{eq:totalrt}) gives the total nearly-linear running time
    $\Otil(m + (n + \cen)\eps^{-2})$.
\end{proof}

\subsection{Incorporating Fast Sequential Edge Additions into the Greedy Algorithm}

We now incorporate $\AddAbove$ into Algorithm~\ref{alg:GreedyTh} to obtain our
nearly-linear time greedy algorithm.
To estimate the maximum effective resistance at Line~\ref{line:ermax}
of Algorithm~\ref{alg:GreedyTh},
we will also need the effective resistance estimation routine from~\cite{DKP+17}:

\begin{lemma}\label{lem:EREst}
    There is a routine $(\hat{r}_{u,v})_{(u,v)\in \ces} = \EREst(G = (V,E), w, \ces, \eps)$ which takes
    a graph $G = (V,E)$ with $n$ vertices and $m$ edges,
    an edge weight function $w : E \to \mathbb{R}^+$,
    a set $\ces$ of $\cen$ vertex pairs,
    and a real number $0 < \eps \leq 1/2$,
    and returns $\cen$ real numbers $(\hat{r}_{u,v})_{(u,v)\in \ces}$
    in $\Otil(m + (n + \cen)\eps^{-2})$ time. With high probability,
    the following statement holds for all $(u,v)\in \ces$:
    \begin{align*}
        (1 - \eps) \er(u,v) \leq \hat{r}_{u,v} \leq (1 + \eps) \er(u,v).
    \end{align*}
\end{lemma}

We give our greedy algorithm in Algorithm~\ref{alg:GreedyThFast}.

\quad

\begin{algorithm}[H]
    \caption{$P = \NSTMaximize(G, \ces, w, \eps, k)$}
	\label{alg:GreedyThFast}
	\Input{
        $G = (V,E)$: A connected graph. \\
        $\ces$: A candidate edge set with $\sizeof{\ces} = \cen$. \\
        $w : (E\union \ces) \to \mathbb{R}^+$: An edge weight function. \\
        $\eps$: An error parameter. \\
        $k$: Number of edges to add.
	}
	\Output{
        $\ses$ : A subset of $\ces$ with at most $k$ edges.
	}
    $P \gets \emptyset$\;
    $(\hat{r}_{u,v})_{(u,v)\in \ces} \gets \EREst(G = (V,E), w, \ces, \eps)$\;
    $er_{max} \gets \frac{1 + \eps}{1 - \eps}\cdot \max\nolimits_{(u,v)\in\ces} w(u,v) \cdot \hat{r}_{u,v}$\;
    $th_0\gets \log (1 + er_{max})$\;
    $th \gets th_0$\;
    \While{$th \geq \frac{\eps}{2\cen} th_0$}{
        Pick an arbitrary ordering $(u_i,v_i)_{i=1}^{q - \sizeof{P}}$
        of edges in $\ces \setminus \ses$. \;
        $\ses' \gets \AddAbove(G, (u_i,v_i)_{i=1}^{q - \sizeof{P}}, w, 2^{th} - 1, \eps/12, k - \sizeof{P})$ \;
        Update the graph by $G \gets G + \ses'$.\;
        $\ses \gets \ses \union \ses'$\;
        $th \gets (1 - \eps/6) th$
    }
    \KwRet{\ses}
\end{algorithm}

\quad

The performance of algorithm $\NSTMaximize$ is characterized in Theorem~\ref{thm:DetMaximize}.

\begin{proof}[Proof of Theorem~\ref{thm:DetMaximize}]
    By Lemma~\ref{lem:ds}, the running time equals
    \begin{align*}
        O(\eps^{-1} \log \frac{\cen}{\eps}) \cdot \Otil(m + (n + \cen)\eps^{-2}) =
        \Otil((m + (n + \cen)\eps^{-2})\eps^{-1}).
    \end{align*}

    We next prove the correctness of the approximation ratio.

    We first consider the case in which the algorithm selects exactly $k$ edges from $\ces$.
    When the algorithm selects an edge $(u_i,v_i)$ at threshold $th$,
    the effective resistance between any $u_j,v_j$ where $i < j\leq \cen$
    can be upper bounded by
    \begin{align}
        \log(1 + \er(u_j, v_j)) &\leq \frac{1}{1 - \eps/6}\cdot \frac{th}{1 - \eps/6} \notag\\
        &\leq \frac{(1 + \eps/6)\log(1 + \er(u_i,v_i))}{(1 - \eps/6)^2} \notag\\
        &\leq \frac{\log(1 + \er(u_i,v_i))}{1 - \eps/2},
        \label{eq:i<j}
    \end{align}
    and the effective resistance between any $u_j,v_j$ where $1 \leq j < i$
    can be upper bounded by
    \begin{align}
        \log(1 + \er(u_j, v_j)) &\leq \frac{1}{1 - \eps/6}\cdot th \notag\\
        &\leq \frac{(1 + \eps/6)\log(1 + \er(u_i,v_i))}{(1 - \eps/6)} \notag\\
        &\leq \frac{\log(1 + \er(u_i,v_i))}{1 - \eps/2}.
        \label{eq:j<i}
    \end{align}
    Thus, the algorithm always picks an edge with at least $1 - \eps/2$ times
    the maximum marginal gain.
    Let $G^{(i)}$ be the graph after the first $i$ edges are added,
   and let $O$ be the optimum solution.
    By submodularity we have, for any $i\geq 0$
    \begin{align*}
        \log \frac{\Tcal(G^{(i+1)})}{\Tcal(G^{(i)})} 
        \geq
        \frac{1 - \eps/2}{k} \cdot \log \frac{\Tcal(G + O)}{\Tcal(G^{(i)})}.
    \end{align*}
    Then, we have
    \begin{align*}
        \log \frac{\Tcal(G^{(k)})}{\Tcal(G)} &\geq
        \kh{1 - \kh{1 - \frac{1 - \eps/2}{k}}^k }\cdot \log \frac{\Tcal(G + O)}{\Tcal(G)} \\
        %&= \kh{1 - \kh{1 - \frac{1 - \eps/2}{k}}^{\frac{k}{1-\eps/2}\cdot (1 - \eps/2)} }
        %\cdot \log \frac{\Tcal(G + O)}{\Tcal(G)} \\
        &\geq \kh{1 - \frac{1}{e^{1 - \eps/2}}} \cdot \log \frac{\Tcal(G + O)}{\Tcal(G)} \\
        &\geq \kh{1 - \frac{1}{e} - \eps/2} \cdot \log \frac{\Tcal(G + O)}{\Tcal(G)},
    \end{align*}
    where the second inequality follows from $(1 - 1/x)^x \leq 1/e$.

    When the algorithm selects fewer than $k$ edges,
    we know by the condition of the while loop that
    not selecting the remaining edges only causes a loss of $(1 - \eps/2)$.
    Thus, we have
    \begin{align*}
        \log \frac{\Tcal(G^{(k)})}{\Tcal(G)} &\geq 
        (1 - \eps/2) \kh{1 - \frac{1}{e} - \eps/2} \cdot \log \frac{\Tcal(G + O)}{\Tcal(G)} \\
        &\geq \kh{1 - \frac{1}{e} - \eps} \cdot \log \frac{\Tcal(G + O)}{\Tcal(G)},
    \end{align*}
    which completes the proof.
\end{proof}

%\thmalgo*

\section{Exponential Inapproximability}
\label{hardness::sec}
In this section, we will make use of properties of some special 
classes of graphs. Amongst these are:  the star graph $S_n$,  which is an $(n+1)$-vertex tree in which $n$ leaves are directly 
connected to a central vertex; an $n$-vertex path graph $P_n$; an $n$-vertex cycle $C_n$;  a fan graph $F_n$, defined by $S_n$ plus a 
$P_n$ supported on its leaves; and a wheel graph $W_n$,  defined by $S_n$ 
plus a $C_n$ supported on its leaves. 
%In this section, we will make use of properties of some special 
%classes of graphs, amongst them are  the star graph $S_n$, the path graph 
%$P_n$, the cycle $C_n$, the fan graph $F_n$, and the wheel graph $W_n$. A star 
%graph $S_n$ is an $n+1$-vertex tree in which $n$ leaves are directly 
%connected to a central vertex. $P_n$ is an $n$-vertex path and $C_n$ 
%is an $n$-vertex cycle. $F_n$ is a graph defined by $S_n$ plus a 
%$P_n$ supported on its leaves. $W_n$ is a graph defined by $S_n$ 
%plus a $C_n$ supported on its leaves. 
We define the length of a 
path or a cycle by the sum of edge weights in the path/cycle. 
Specifically, the lengths of $P_n$ and $C_n$ with all weights equal to $1$
are $n-1$ and $n$, respectively. We also write $K_n$ to denote the $n$-vertex complete graph.
%\yhy{Added the definition form $K_n$}

\subsection{Hardness of Approximation for Minimum Path Cover}

We begin by introducing the definition of the minimum path cover problem.

\begin{problem}[Minimum Path Cover]\label{prob:mpc}
Given an  undirected graph $G=(V,E)$, a path cover is a set of disjoint paths such that every vertex $v\in V$ belongs to exactly one path. Note that a path cover may include paths of length 0 (a single vertex). The minimum path cover problem is to find a path cover of G having the least number of paths.
\end{problem}

We  recall a known hardness result of TSP with distance $1$ and $2$ ($\rm{(1,2)}$-TSP).
%\todo{define length of cycles and paths, perhaps in preliminary.}
\begin{lemma}[\cite{PY93,EK01}]\label{12TSP}
    There is a constant $\sigma > 0$, such that it is $\mathbf{NP}$-hard
    to distinguish between the instances of {\rm{$(1,2)$-TSP}} ($K_n$ with edge weights $1$ and $2$)
    having shortest Hamiltonian cycle length $n$ and shortest Hamiltonian cycle length at least $(1 + \sigma n)$.
    %In a {\rm{$(1,2)$-TSP}} instance (a $K_n$ with edge weights $1$ and $2$),
    %there is a constant $\sigma > 0$ such that it is $\mathbf{NP}$-hard to
    %determine whether
    %there is a Hamiltonian cycle of length $n$ or any Hamiltonian cycle in the graph has a length
    %at least $1+\sigma n$.
\end{lemma}
Next, we reduce the {\rm{$(1,2)$-TSP}} problem to a TSP-Path problem with distance $1$ and $2$ ($\rm{(1,2)}$-TSP-Path).
\begin{lemma}
    There is a constant $\delta > 0$, such that it is $\mathbf{NP}$-hard
    to distinguish between the instances of {\rm{$\rm{(1,2)}$-TSP-Path}} ($K_n$ with edge weights $1$ and $2$)
    having shortest Hamiltonian path length $(n - 1)$ and shortest Hamiltonian path length at least
    $(1 + \delta n)$.
%In a {\rm{$\rm{(1,2)}$-TSP-Path}} instance, there is a constant $\delta > 0$ such that it's $\mathbf{NP}$-hard to distinguish whether there is a Hamiltonian path of length $n-1$ or any Hamiltonian path in the graph has a length at least $(1+\delta)n$.
\end{lemma}
\begin{proof}
%The completeness follow directly from Lemma~\ref{12TSP}.
\emph{Completeness}: If a (1,2)-weighted complete graph $K_n$ 
has a Hamiltonian cycle of length $n$, then it has a Hamiltonian path of length $n-1$.

\emph{Soundness}: Given a (1,2)-weighted complete graph $K_n$, let $\mathcal{Q}$ be 
a shortest Hamiltonian cycle in $K_n$, and suppose its length satisfies $\sizeof{\mathcal{Q}}\geq (1+\sigma)n$. 
Then the shortest Hamiltonian path in $K_n$ has length at least
$(1+\sigma)n-2$. Let $\delta = \frac{\sigma}{2}$. Then for $n > \frac{4}{\sigma}$,
$\sizeof{\mathcal{H}} \geq (1+\delta)n$.
%\todo{soundness is also easy}
\end{proof}

\begin{lemma}
There is a constant $\delta$ such that it is $\mathbf{NP}$-hard to
distinguish between graphs with minimum path cover number $1$ and minimum
path cover number at least $\delta n$.
%and whether there is a path cover of size $1$ or the minimum path
% cover has a size at least $\delta n$.
\begin{proof}
\emph{Completeness}: In a (1,2)-weighted complete graph $K_n$, if $\mathcal{H}^*$ is 
a shortest Hamiltonian path in $K_n$, and its length satisfies $\sizeof{\mathcal{H}^*}=n-1$, 
then $\mathcal{H}^*$ is a Hamiltonian path of the subgraph consists of 
all edges in $K_n$ that have weights equal to $1$.

\emph{Soundness}: In a (1-2)-weighted complete graph $K_n$, in which any 
Hamiltonian path $\mathcal{H}$ satisfies $\sizeof{\mathcal{H}} \geq (1+\delta)n$, 
the minimum path cover number in the subgraph consists of all edges in $K_n$ 
with weight $1$ is at least $\delta n+2$.
\end{proof}
\end{lemma}

\subsection{Exponential Inapproximability for NSTM}

We now consider an instance of NSTM we described in Section~\ref{sec:res}.
	Recall that we obtain the number of spanning trees in $G+P$ by iteratively using the matrix 
	determinant lemma: $$\Tcal(G^{(i+1)}) = \kh{ 1+\er^{G^{(i)}}(u_i,v_i) } \Tcal(G^{(i)})\,. $$
	We are interested in the effective resistance between endpoints of an edge $e\in\ses$ 
	that is not contained in any path in a minimum path cover $\Pcal$ of the graph $H[V']$ using
	only edges in $\ses$.
	
	\begin{lemma}
	\label{intterRes:lemma}
		Let $H = (V, E)$ be an unweighted graph equal to a star $S_n$ plus $H$'s subgraph 
		$H[V']=(V', E')$ supported on $S_n$'s leaves (see Figure~\ref{starPlusH:fig}).
        Let
        $P$ be an arbitrary subset of $E'$ with $\sizeof{P} = n - 1$,
        and $\Pcal$ be the minimum path cover of $H[V']$
        using only edges in $\ses$.
        Let $e=(u_i, v_i)$ be an edge that is contained
		in $\ses$ but not in the path cover $\Pcal$, then:
		\begin{enumerate}
            \item If $u_i$ is not an endpoint of a path in the path cover $\Pcal$,
			and $v_i$  is not an isolated vertex in $\Pcal$, then
            \[
                \er(u_i, v_i)\leq \frac{7}{6}.
            \]
			\item If $u_i$  is not an endpoint of a path in the path cover $\Pcal$,
			and $v_i$ is an isolated vertex in $\Pcal$, then
            \[
                \er(u_i, v_i) \leq \frac{3}{2}.
            \]
			\item If $u_i$ and $v_i$ are endpoints of the same path in the path cover $\Pcal$,
			then
            \[
                \er(u_i, v_i) < \sqrt{5}-1.
            \]
		\end{enumerate}
	\end{lemma}
	\begin{figure}
        \centering
  	\begin{subfigure}[b]{0.35\textwidth}
        \centering
    	\includegraphics[width=\textwidth]{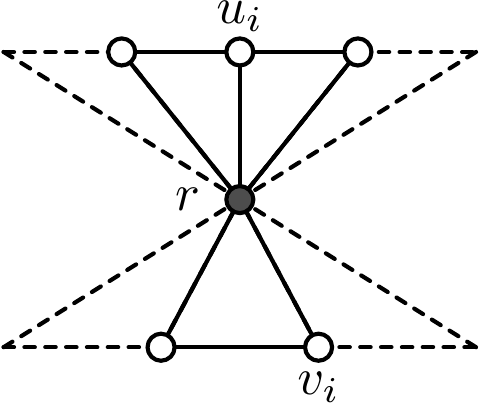}
    	\subcaption{Case 1}
    	\label{subfig:1}
  	\end{subfigure}
  	\qquad
  	\begin{subfigure}[b]{0.35\textwidth}
        \centering
    	\includegraphics[width=\textwidth]{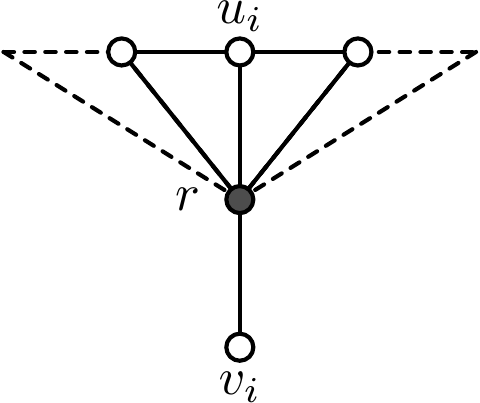}
    	\subcaption{Case 2}
    	\label{subfig:2}
  	\end{subfigure}
    \caption{Explanation for Lemma~\ref{intterRes:lemma}. Dotted lines represent
    remaining parts of the path that $u_i$ or $v_i$ belongs to, and the edges
    connecting $r$.}
  	\label{ERBetwPaths:fig}
	\end{figure}
	\begin{proof}
		From Rayleigh's Monotonicity law, we know that for the first case 
		\[
            \er(u_i, v_i)\leq \frac{1}{2}+\frac{2}{3}.
        \]
        Similar analysis shows
		that 
        \[
            \er(u_i, v_i)\leq \frac{1}{2}+1
        \]
        holds for the second case in the lemma.
        Figure~\ref{ERBetwPaths:fig} shows how we obtain these bounds.
		
		The last case follows directly from the fact that
        \[
            \Tcal(W_n) < \sqrt{5} \cdot \Tcal(F_n),
        \]
		which was shown in~\cite{MEM14}.
	\end{proof}		
	
	We first prove the completeness of the reduction.
\begin{lemma}
	Let $H = (V, E)$ be an unweighted graph equal to a star $S_n$ plus its subgraph
	$H[V']=(V', E')$ supported on $S_n$'s leaves.
    In the NSTM instance where $G=S_{n}$, $\ces=E'$, and
    $k=n-1$, there exists $\ses\subseteq \ces$, $|\ses|=k$ which satisfies
    $$\Tcal\kh{G+\ses} = \frac{1}{\sqrt{5}}\kh{\kh{\frac{3+\sqrt{5}}{2}}^{n}
    -\kh{\frac{3-\sqrt{5}}{2}}^{n}}\,$$
    if $H[V']$ has a Hamiltonian path.
\end{lemma}

\begin{proof}
    Suppose there is a Hamiltonian path $\Pcal^\ast$ in graph $H[V']$, then $S_{n}+ \Pcal^\ast$ 
    is a fan $F_n$. It's number of
    spanning trees is given in~\cite{MEM14} as
    \begin{align}
        \Tcal(S_{n}+ \Pcal^\ast)= \Tcal(F_{n}) = \frac{1}{\sqrt{5}}\kh{\kh{\frac{3+\sqrt{5}}{2}}^{n}
        -\kh{\frac{3-\sqrt{5}}{2}}^{n}}\,.
    \end{align}

%in $G_0$ that do not constitutes a Hamiltonian path,
%we can find a sequence of edge sets $E_1, E_2,\dots, E_{n}$ which maintains the same cardinality while
%$\Tcal(G_{\rm{ini}}\cup E_0) < \Tcal(G_{\rm{ini}}\cup E_1)<\dots<\Tcal(G_{\rm{ini}}\cup E_n)$ and
%$E_n$ is a Hamiltonian path of $G_0$.
%
%Before we prove the necessity we introduce a lemma which compares number of spanning trees in some considered graphs.
\end{proof}

Before proving the soundness of the reduction, we warm up by proving the following lemma:

\begin{lemma}
	Let $H = (V, E)$ be an unweighted graph equal to a star $S_n$ plus its subgraph
	$H[V']=(V', E')$ supported on $S_n$'s leaves.
    In the NSTM instance where $G=S_{n}$, $\ces=E'$, and
    $k=n-1$, for any edge 
    set $\ses\subseteq \ces$ with $k$ edges that does not constitute a 
    Hamiltonian path, $$\Tcal\kh{S_{n}+\ses} < \Tcal\kh{S_{n}+\Pcal^\ast}.$$ 
\end{lemma}

\begin{proof}

    Suppose $\Pcal'=\{\pth'_1, \pth'_2,\dots, \pth'_{t'} \}$ is a minimum path cover of $G''=(V', P)$,
    with $\sizeof{\pth_i'} = l'_i$ and $l'_1 \leq l'_2 \leq \dots \leq l'_{t'} $. Suppose 
    $l'_1 = \dots = l'_\kappa = 0< l'_{\kappa+1}$, then
    \begin{align}
        \Tcal\kh{S_{n} + \Pcal'} & = \prod_{i=1}^{t'} \frac{1}{\sqrt{5}}\kh{\kh{\frac{3+\sqrt{5}}{2}}^{l'_i+1}
        -\kh{\frac{3-\sqrt{5}}{2}}^{l'_i+1}}\,,
    \end{align}
    %\todo{discuss isolated vertices}
    Since $7/6 < \sqrt{5}-1 < 3 / 2$
    \begin{align}
        \Tcal\kh{S_{n} + P} &\leq \kh{\frac{5}{2}}^\kappa\cdot 5^{(t-\kappa-1)/2}\cdot 
        5^{-(t-\kappa)/2} \prod_{i=\kappa+1}^{t'}
        \kh{\kh{\frac{3+\sqrt{5}}{2}}^{l'_i+1}-\kh{\frac{3-\sqrt{5}}{2}}^{l'_i+1}}     \label{snp:ieqty1} \\
        & <\kh{\frac{5}{2}}^\kappa \cdot \frac{1}{\sqrt{5}}\kh{\kh{\frac{3+\sqrt{5}}{2}}^{n-\kappa} 
        -\kh{\frac{3-\sqrt{5}}{2}}^{n-\kappa}}\qquad \label{snp:ieqty1:2}\\
        & <\frac{1}{\sqrt{5}}\kh{\kh{\frac{3+\sqrt{5}}{2}}^{n}
        -\kh{\frac{3-\sqrt{5}}{2}}^{n}}  \notag\\
        & = \Tcal(S_{n+1}+ \Pcal^\ast)\,. \notag
    \end{align}
    where (\ref{snp:ieqty1:2}) follows from (\ref{snp:ieqty1}) by the fact that $\sum_{i=\kappa+1}^{t'}\kh{l'_i+1}=n-\kappa$.
\end{proof}

Next, we prove the soundness of the reduction.
\begin{proof}[Proof of Lemma~\ref{soundness:lemma}]
    Suppose the minimum path cover number of $H[V'] = (V',E')$ is at least $\delta n$.
	%Let $\Pcal$ be the minimum path cover of $G'' =(V',\ses)$
	%and $t\geqslant \delta n$.
    Then any path cover $\Pcal'=\{\pth'_1, \pth'_2,\dots, \pth'_{t'} \}$
	of $G''=(V', \ses)$ must satisfy $t'\geqslant t$. We let $\sizeof{\pth'_i}=l'_i$
	for $i\in \{1,\dots, t'\}$ and $l'_1\leq l'_2 \leq \dots \leq l'_{t'} $. Suppose 
    $l'_1 = \dots = l'_\kappa = 0 < l'_{\kappa+1}$; then according to (\ref{snp:ieqty1}), we obtain
	\begin{align*}
        \Tcal\kh{S_{n} + P} &\leq \kh{\frac{5}{2}}^{\kappa}\cdot \frac{1}{\sqrt{5}} \prod_{l'_i\geqslant 1}
        \kh{\kh{\frac{3+\sqrt{5}}{2}}^{l'_i+1}-\kh{\frac{3-\sqrt{5}}{2}}^{l'_i+1}}
    \end{align*}
    Since the average length of paths in $\Pcal'$ satisfies
    $$\expec{i}{l'_i}\leqslant \frac{n-\delta n}{\delta n}=\frac{1-\delta}{\delta},$$
    by Markov's inequality, $$\prob{i}{l'_i \geqslant \frac{2(1-\delta)}{\delta}}
    \leqslant \frac{1}{2}\,.$$ Therefore %\todo{discuss isolated vertices}
    \begin{align*}
    	\Tcal\kh{S_{n} + P} \leq &\kh{\frac{5}{2}}^{\kappa}\cdot \frac{1}{\sqrt{5}}
    	\kh{\prod_{1 \leq l'_i<\frac{2(1-\delta)}{\delta}}
    	\kh{\kh{\frac{3+\sqrt{5}}{2}}^{l'_i+1}-\kh{\frac{3-\sqrt{5}}{2}}^{l'_i+1}}}\\
    	&\qquad \cdot \kh{\prod_{l'_i\geqslant \frac{2(1-\delta)}{\delta}}
    	\kh{\kh{\frac{3+\sqrt{5}}{2}}^{l'_i+1}-\kh{\frac{3-\sqrt{5}}{2}}^{l'_i+1}}}\\
    	\leq &\kh{\frac{5}{2}}^{\kappa} \cdot \frac{1}{\sqrt{5}}
    	\kh{ \prod_{1\leq l'_i<\frac{2(1-\delta)}{\delta}}
    	\kh{ 1- \kh{\frac{3-\sqrt{5}}{2}}^{2l'_i+2} }}
    	\kh{\frac{3+\sqrt{5}}{2}}^{n-\kappa}\\
    	<& \kh{\frac{5}{2}}^\kappa\cdot \frac{1}{\sqrt{5}}\kh{ 1- 
    	\kh{\frac{3-\sqrt{5}}{2}}^{\frac{4-2\delta}{\delta}} }
    	^{\frac{\delta n}{2}-\kappa}
    	\kh{\frac{3+\sqrt{5}}{2}}^{n-\kappa} \\%\quad \text{By Markov's inequality}\\
    	=& \kh{\frac{5}{2}}^\kappa\cdot \frac{1}{\sqrt{5}}\kh{ 1- 
    	\kh{\frac{3-\sqrt{5}}{2}}^{\frac{4-2\delta}{\delta}} }
    	^{\frac{\delta n}{2}-\kappa}
    	\kh{\frac{3+\sqrt{5}}{2}}^{\frac{\delta n}{2}-\kappa} 
    	\kh{\frac{3+\sqrt{5}}{2}}^{n-\frac{\delta n}{2}}\\
    	=& \frac{1}{\sqrt{5}}\kh{\frac{5}{2}}^{\kappa}
    	\kh{ \kh{\frac{3+\sqrt{5}}{2}} - \kh{\frac{3-\sqrt{5}}{2}}^{\frac{4-3\delta}{\delta}} }
    	^{\frac{\delta n}{2}-\kappa}
    	\kh{\frac{3+\sqrt{5}}{2}}^{-\frac{\delta}{2}n}
    	\kh{\frac{3+\sqrt{5}}{2}}^{n}\,.
    \end{align*}
    The third inequality follows by Markov's inequality. Let
    \begin{align}
    	\alpha \defeq \max \setof{ \kh{\frac{15-5\sqrt{5}}{4}}^{\frac{\delta}{2}}\, , \,\,\,
    	\kh{1
    	- \kh{\frac{3-\sqrt{5}}{2}}^{\frac{4-2\delta}{\delta}}}^{\frac{\delta}{2}} }\,.
    \end{align}
    Since $\delta$ is a positive constant, $\alpha$ is a constant that satisfies $0<\alpha<1$. Then
    \begin{align*}
    	\Tcal\kh{S_{n} + P} < \frac{1}{\sqrt{5}}\kh{\alpha \cdot \frac{3+\sqrt{5}}{2}}^n\,.
    \end{align*}
    Therefore
    \begin{align*}
    	\frac{\log \Tcal\kh{S_{n}+P} }{\log \Tcal \kh{F_n}}& <\frac{-\frac{1}{2} \log 5+ n\cdot
    	 \log \kh{\alpha\cdot \frac{3+\sqrt{5}}{2}}}{-\frac{1}{2} \log 5+  n\cdot 
    	 \log \kh{\frac{3+\sqrt{5}}{2}} + \log \kh{ 1- \kh{\frac{3-\sqrt{5}}{2}}^{2n}}}\\
    	& < \frac{n\log \kh{\alpha \cdot \frac{3+\sqrt{5}}{2}}}
    	{n\log \kh{\frac{3+\sqrt{5}}{2}}+ \log \kh{ 1- \kh{\frac{3-\sqrt{5}}{2}}^{2n}}} \,.
    \end{align*}
    If $n$ satisfies $$n>\max \setof{\frac{1}{2\log\frac{3+\sqrt{5}}{2}}\,,\,
    \frac{2}{\log \frac{1}{\alpha}} }\,,$$
    then
    \begin{align*}
    	\frac{\log \Tcal\kh{S_{n}+P} }{\log \Tcal \kh{F_n}}& 
    	<\frac{ n\cdot
    	 \log \kh{\alpha}+n\cdot \log\kh{ \frac{3+\sqrt{5}}{2}}}{ n\cdot 
    	 \log \kh{\frac{3+\sqrt{5}}{2}} + \log \kh{ 1- \kh{\frac{3-\sqrt{5}}{2}}^{2n}}}\\
    	& < \frac{\frac{1}{2}\log \kh{\alpha} +\log\kh{ \frac{3+\sqrt{5}}{2}}}
    	{\log \kh{\frac{3+\sqrt{5}}{2}}}\\ 
    	& = 1 - \frac{\log \kh{1/ \alpha}}{2\log \kh{\frac{3+\sqrt{5}}{2}}}\,.
    \end{align*}
    Thus, we obtain the constant $c$ stated in Lemma~\ref{soundness:lemma}:
    \begin{align}
    	c = \frac{\log \kh{1/ \alpha}}{2\log \kh{\frac{3+\sqrt{5}}{2}}}\,.
    \end{align}
    %\todo{find the absolute constant}
\end{proof}

\newpage

\newcommand{\etalchar}[1]{$^{#1}$}

%\bibliographystyle{alpha}
%\bibliography{ref,sep_ref}

%\begin{appendix}
%    \input{append}
%\end{appendix}

\end{document}